\documentclass{article}
\usepackage[utf8]{inputenc}
\usepackage{fullpage}
\usepackage{amsmath,amssymb,amsthm,mathrsfs}
\usepackage{braket,physics}
\usepackage{complexity}
\usepackage[ruled,linesnumbered,noend]{algorithm2e}
\usepackage{setspace}
\usepackage[noend]{algpseudocode}
\usepackage{parskip}
\usepackage[colorlinks,linkcolor=Red,citecolor=Blue]{hyperref}
\usepackage[capitalize]{cleveref}
\usepackage[dvipsnames]{xcolor}
\usepackage{subcaption}
\usepackage{tikz}
\usetikzlibrary{quantikz}

\usepackage[page,header]{appendix}
\usepackage{titletoc}

\title{Pauli Manipulation Detection codes and Applications to Quantum Communication over Adversarial Channels}
\author{Thiago Bergamaschi\thanks{UC Berkeley. Email: \texttt{thiagob@berkeley.edu}}}
\date{\today}

\newif\ifnotes
\notestrue
\newcommand{\Enc}{\mathsf{Enc}}
\newcommand{\Auth}{\mathsf{Auth}}
\newcommand{\Dec}{\mathsf{Dec}}
\newcommand{\PMD}{\mathsf{PMD}}
\newcommand{\Aux}{\mathsf{Aux}}
\newcommand{\Anc}{\mathsf{Anc}}
\newcommand{\PTC}{\mathsf{PTC}}
\newcommand{\NM}{\mathsf{NM}}

\newcommand{\Acc}{\mathsf{Acc}}
\newcommand{\Rej}{\mathsf{Rej}}

\newcommand{\suppress}[1]{}

\newtheorem{theorem}{Theorem}[section]

\newtheorem{lemma}[theorem]{Lemma}
\newtheorem{corollary}[theorem]{Corollary}
\newtheorem{claim}[theorem]{Claim}
\newtheorem{definition}{Definition}[section]

\newtheorem{fact}{Fact}[section]
\newtheorem{task}{Task}

\newcommand{\comment}[1]{}

\renewcommand{\epsilon}{\varepsilon}

\SetKwInput{KwInput}{Input}

\SetKwInput{KwOutput}{Output}

\begin{document}
\maketitle

\begin{abstract}


We introduce and explicitly construct a quantum error-detection code we coin a ``Pauli Manipulation Detection” code (or PMD), which detects every Pauli error with high probability. We apply them to construct the first near-optimal codes for two tasks in quantum communication over adversarial channels.

Our main application is an approximate quantum code over qubits which can efficiently correct from a number of (worst-case) erasure errors approaching the quantum Singleton bound. Our construction is based on the composition of a PMD code with a stabilizer code which is list-decodable from erasures, a variant of the stabilizer list-decodable codes studied by \cite{Leung2006CommunicatingOA, Bergamaschi2022ApproachingTQ}.

Our second application is a quantum authentication code for ``qubit-wise'' channels, which does not require a secret key. Remarkably, this gives an example of a task in quantum communication which is provably impossible classically. Our construction is based on a combination of PMD codes, stabilizer codes, and classical non-malleable codes \cite{Dziembowski2010}, and achieves ``minimal redundancy" (rate $1-o(1)$).



\end{abstract}

\pagenumbering{arabic}
\newpage

\setcounter{tocdepth}{2}
{\small\tableofcontents}
\newpage

\section{Introduction}



Algebraic Manipulation Detection (AMD) codes, introduced by Cramer, Dodis, Fehr, Padr\'o and Wichs \cite{Cramer2008DetectionOA}, are a fundamental primitive at the intersection of coding theory and cryptography. They are a form of keyless message authentication code, which offers error-detection guarantees against additive errors by assuming that the tampering adversary cannot “see” the codeword. In other words, AMDs detect arbitrary bit-flip errors with high probability, so long as the error is independent of the internal randomness in the code. Although their error-detection guarantees may seem a bit restrictive at first, AMD codes have found numerous applications as building blocks to stronger cryptographic primitives, including secret sharing schemes \cite{Cramer2015LinearSS}, non-malleable codes \cite{Dziembowski2010}, fuzzy extractors \cite{Cramer2008DetectionOA}, and more.\footnote{We refer the reader to \cite{Cramer2013AlgebraicMD} for a review of AMD codes, and \cref{section:related} for more related work.}

\textbf{In this paper} we consider a quantum analog of AMD codes. In particular, we introduce and explicitly construct quantum error-detection codes - subspaces of complex vector spaces - that can detect arbitrary Pauli errors (of arbitrarily large weight). This error model includes both bit-flip and phase-flip errors, generalizing in a sense their classical counterparts. Inspired by their name, we refer to our codes as “Pauli Manipulation Detection” codes (or PMDs):\footnote{$\|\cdot\|_\infty$ is the operator norm, or Schatten infinity norm. We refer the reader to \Cref{def:PMD2} for a formal definition, and \cref{section:prelim} for background on Pauli operators and matrix norms.}


\begin{definition} 
We refer to a subspace $\Pi$ of an $n$-qubit Hilbert space as an $\epsilon$-$\PMD$ if, for every $n$-qubit Pauli error $E\neq \mathbb{I}$, 
\begin{equation}
    \|\Pi E\Pi\|_\infty\leq \epsilon.
\end{equation}
\end{definition}

In the (informal) definition above, $\Pi$ plays both the role of code-space, and the projection onto it. This definition implies that every Pauli error on the code is detected \textit{almost}-\textit{deterministically}: if a code-state $\ket{\psi}\in \Pi$ is corrupted by a Pauli error $E\neq \mathbb{I}$, the projective measurement onto the PMD code-space $(\Pi, \mathbb{I}-\Pi)$ does not detect the Pauli tampering with probability at most
\begin{equation}
 \bra{\psi}E^\dagger \Pi E \ket{\psi} \leq \|\Pi E^\dagger \Pi E\Pi \|_\infty = \|\Pi E\Pi\|^2_\infty\leq \epsilon^2.
\end{equation}
Crucially, the gentle-measurement lemma then implies that PMD codes are able to detect whether they were corrupted by a Pauli, without disturbing the corrupted code-state. That is, conditioned on measuring $\mathbb{I}-\Pi$,
\begin{equation}
    (\mathbb{I}-\Pi) E \ket{\psi} \approx E\ket{\psi}. 
\end{equation}
We leverage this key idea to construct the first efficient quantum codes approaching the information-theoretically optimal rate (or ``minimal redundancy") for two tasks in quantum coding theory and cryptography.

\textbf{Applications} Our main application is an approximate quantum code over qubits for the adversarial erasures channel, which is able to efficiently correct from a near-optimal number of erasure errors. Due to a well-known connection between quantum error correction and secret sharing \cite{Gottesman1999TheoryOQ, Cleve1999HOWTS, Smith2000QuantumSS, Crpeau2005ApproximateQE}, this result can also be understood as a near-optimal \textit{ramp} quantum secret sharing scheme with qubit shares. 

Our second application is a quantum \textit{tamper-detection} code for ``qubit-wise" channels. In turn, this result can be thought of as a quantum authentication code for a restricted, un-entangled, adversarial model, which does not require a secret key. Curiously, these codes provide a form of error-detection which is provably impossible with classical messages.

\subsection{Our Results}

\subsubsection{Constructions of PMD codes} 

Our first result is an explicit construction of PMD codes, which achieves rate near $1$ and inverse-exponential error. 

\begin{theorem}\label{theorem:results-pmdexplicit}
    For all sufficiently large integers $n$ and $\lambda | n$, there exists an $\epsilon$-$\PMD$ encoding $n-\lambda$ qubits into $n$ qubits with error $\epsilon \leq n^{1/2}\cdot 2^{1-\lambda/4}$. Moreover, it can be constructed and encoded efficiently.
\end{theorem}

Our construction is based on the stabilizer ``Purity Testing" codes (PTCs) introduced by \cite{barnum2002authentication} in the context of quantum message authentication. Informally, a PTC is a pseudorandom set of stabilizer codes\footnote{Stabilizer codes \cite{Gottesman1997StabilizerCA} are a quantum analog of linear codes, defined by the joint $+1$ eigenspace of a set of commuting Pauli operators (the generators). A Pauli error is said to be detectable if it anti-commutes with a generator, see \cref{definition:stabilizercode}. } $\{Q_k\}_{k\in K}$ which detects every Pauli error with high probability (over random choice of the key $k$). We show that the natural approach of encoding a message state $\ket{\psi}$ into a superposition of PTC encodings,
\begin{equation}
    \Enc_\PMD\ket{\psi} = |K|^{-1/2}\sum_{k\in K}\ket{k}\otimes \Enc_{Q_k}\ket{\psi},
\end{equation}
\textit{almost} defines a PMD.\footnote{$\Enc_\PMD$ and the $\Enc_{Q_k}$ are unitaries, acting on a message register and ancilla qubits initialized to $\ket{0}$ (omitted for clarity).} Unfortunately, in general this construction is not a PMD. Nevertheless, we show that the family of PTC's designed in \cite{barnum2002authentication} satisfy a certain form of ``key manipulation security", which guarantees this construction is secure (and maybe of independent interest).

\subsubsection{Approximate Quantum Erasure Codes on the Quantum Singleton Bound}
\label{section:aqec-result}

Our first application of PMD codes is to approximate quantum error correction.

An erasure error (or detectable leakage) on a quantum error correcting code corresponds to an error on a known location of the code. Erasure correction has recently found renewed interest in the quantum computing community, due to connections with quantum secret sharing \cite{Cleve1999HOWTS}, fault tolerance \cite{Wu2022ErasureCF}, the information-disturbance tradeoff \cite{Kretschmann2006TheIT}, and even information retrieval in black holes \cite{Yoshida2017EfficientDF}.

The quantum Singleton or ``no-cloning" bound imposes an information-theoretic limit on the maximum number of erasures a quantum code can (even approximately) correct from. It states that no quantum code of rate $r$ can correct from more than a $\frac{1}{2}(1-r)+o(1)$ fraction of erasures, for any assymptotically small recovery error \cite{Knill1996TheoryOQ, Rains1997NonbinaryQC, Grassl2020EntropicPO}. Moreover, to approach this bound over small alphabet sizes, one fundamentally needs \textit{approximate} quantum error correction: A famous result by Rains \cite{Rains1996QuantumSE} showed that no exact quantum code on qubits can correct from more than even a $1/3$ fraction of erasures. 

Our main result in this work is a randomized construction of an efficient and near-optimal approximate quantum code over qubits for the adversarial (worst-case) erasures channel.

\begin{theorem}[Main Result]\label{theorem:main}
    For every rate $0<r<1$ and sufficiently large $n\in \mathbb{N}$, there exists a Monte Carlo construction of an $n$ qubit quantum code of rate $r$ which corrects from a fraction $\delta\geq \frac{1}{2}(1-r) - O(\log^{-1} n)$ of erasures, up to an inverse-exponentially small recovery error $\epsilon = 2^{-\Tilde{\Omega}(n)}$. The construction succeeds with probability $1-2^{-\Tilde{\Omega}(n)}$.
\end{theorem}

Moreover, one can encode, decode, and sample a classical description of our codes efficiently in poly$(n)$ time. We describe the relation between \cref{theorem:main} and quantum secret sharing in \cref{section:related}.

A challenging consequence of the relaxation to approximate erasure correction is that non-adaptive (or oblivious) adversaries are fundamentally very different than adaptive adversaries \cite{Lin2018SecretSW}, which can use the erased qubits to learn partial information about the message (and modify their corruption). While against the former one can often borrow ideas from the well-studied random erasure error model \cite{Hayden2017ApproximateQE, Gullans2020QuantumCW}, in the adaptive setting we require very different techniques.

Our code construction is based on a recent approach to design quantum codes against adversarial errors, list-decodable stabilizer codes \cite{Leung2006CommunicatingOA, Bergamaschi2022ApproachingTQ}. However, known approaches to use list decoding for approximate quantum error correction either assume shared randomness between sender and receiver, or require a large qu\textit{d}it alphabet size. Thereby, they are unable to approach the quantum Singleton bound over qubits, at least in a ``one-shot" setting.

Our new insights lie in combining our PMD codes, with a new type of stabilizer code which is list-decodable \textit{from erasures}, an analog to the classical notion introduced by \cite{Guruswami2001ListDF}. In this application, our main technical contribution is a reduction from approximate quantum erasure correction to PMDs and erasure list-decoding, and is inspired by a classical result by \cite{Cramer2015LinearSS} (Eurocrypt 2015) on the construction of classical secret sharing schemes from AMDs and list-decodable codes.  



\subsubsection{Quantum Tamper-Detection codes for Qubit-wise channels}

Our second application of $\PMD$ codes is to designing quantum tamper-detection codes. 

Tamper-detection \cite{Jafargholi2015TamperDA, Boddu2021TamperDA} is an adversarial coding-theoretic guarantee which is closely related to error-detection and authentication, but \textit{without} a secret key. Informally, both quantum tamper-detection codes (TDCs) and quantum authentication schemes (QASs) \cite{barnum2002authentication, Hayden2016TheUC} are protocols in which a sender (Alice) conveys a (quantum) message $\psi$ to a receiver (Bob) over some insecure quantum channel $\Lambda$, with the guarantee that Bob either receives a state close to $\psi$, or aborts $\bot$.\footnote{See \cref{definition:qa} for a formal definition.}

Their distinction lies in the adversarial model: whereas QASs are secure against arbitrary channels $\Lambda$, TDCs only hope to address certain restricted tampering models. Naturally, the stronger adversarial model also comes with a drawback. \cite{barnum2002authentication} showed that, in general, sending an authenticated $k$ qubit message requires a secret key of at least $2k\cdot (1-o(1))$ bits shared between Alice and Bob. The motivating question we ask here is: 

\textit{Are there simpler adversarial models for which quantum tamper-detection is possible without a secret key?}

The model we study in this paper is that of ``qubit-wise" channels, where $\Lambda = \otimes_{i\in [n]}\Lambda_i$. Informally, this corresponds to Alice sending each qubit of her $n$ qubit code-state to $n$ distinct, non-communicating and un-entangled parties. These parties may strategize ahead of time and decide on a nefarious combination of distinct channels $\Lambda_i\neq \Lambda_j$, but they must act on the code-state independently of each other.\footnote{We remark that by an averaging argument, the adversaries are allowed shared randomness, just not shared entanglement.}

It should be clear that this task remains impossible classically. Indeed, in the absence of shared randomness between sender and receiver, the $n$ parties can always collude ahead of time, and substitute the entire message with a valid encoding of any other pre-agreed message $\hat{m}$. Bob will be none the wiser, and will decode to $\hat{m}$ without noticing that Alice's message has been completely replaced. Nonetheless, we show that in the quantum setting the scenario is drastically different:

\begin{theorem}\label{theorem:auth-results}
    For every sufficiently large $n\in \mathbb{N}$, there exists a Monte Carlo construction of a quantum tamper code for qubit-wise channels, of blocklength $n$ qubits, rate $1-O(\log^{-1}n)$, and error $\epsilon \leq 2^{-\Tilde{\Omega}(n)}$. The construction succeeds with probability $1-2^{-\Tilde{\Omega}(n)}$.
\end{theorem}

Although classical tamper detection is impossible in the ``bit-wise" setting, the closely related goal known as \textit{non-malleability} is possible. Classical non-malleable codes \cite{Dziembowski2010} are a relaxation of authentication and error correcting codes, where Bob is tasked with decoding Alice's message $m$ to either the original message $m$, rejection $\bot$, or a completely unrelated value $\Tilde{m}$. In contrast to authentication and error correction, non-malleable codes are much more versatile and can be constructed in a number of limited-adversarial settings, even without a secret key. 

Our code construction combines non-malleable codes with PMD codes and stabilizer codes, and is inspired by the design of a classical bit-wise non-malleable code by \cite{Cheraghchi2013NonmalleableCA}. However, we emphasize that our result is fundamentally stronger than any form of classical or quantum non-malleability (for qubit-wise channels, see \cref{section:overview-keyless}). In this application, our main technical contribution draws ideas from the quantum circuit lower bounds by \cite{Anshu2020CircuitLB}, to show how to use entanglement to evade the previously mentioned ``substitution" attack, and detect adversarial tampering without a secret key.

\subsection{Additional Related Work}
\label{section:related}

\textbf{Tamper-detection codes and Secret Sharing.} Our work is inspired by constructions of classical secret sharing schemes from ideas in tamper-resilient cryptography \cite{Cramer2001OnTC, Cramer2008DetectionOA, Cramer2015OptimalAM, Cramer2015LinearSS, Lin2018SecretSW}. A key ingredient in that line of work is the concept of an algebraic manipulation detection (AMD) code \cite{Cramer2008DetectionOA}. Simply put, an ``algebraic manipulation" is a form of adversarial data tampering, without prior knowledge of the data values. For instance, if our data is a value $x$ in a group $\mathcal{G}$, then an ``algebraic manipulation" corresponds to adding some value $\delta\in \mathcal{G}$ (without knowing $x$), resulting in $x+\delta$. An AMD code corresponds to an encoding of data into $\mathcal{G}$, in a way that any algebraic manipulation $\delta\neq 0$ is detected with high probability. \cite{Jafargholi2015TamperDA} generalized AMD codes to other tampering families, including low-degree polynomial and low-depth circuit tampering.

\cite{Boddu2021TamperDA} were the first to explore tamper-detection codes in the quantum setting. Their work focused on the concept of a ``Unitary Tamper Detection Code", a quantum code which can faithfully detect errors chosen from a restricted family of unitary operators. Among other results, they constructed explicit quantum codes for classical messages which could detect arbitrary Pauli errors, which they referred to as ``Quantum AMDs". Our PMDs arose as a natural generalization to quantum messages, albeit our techniques depart significantly. Moreover, \cite{Boddu2021TamperDA} don't attempt to address tampering \textit{channels} (just unitaries), whose lack of reversibility imply a unique set of challenges. 

\textbf{Approximate Quantum Error Correction.} \cite{Leung1997ApproximateQE} were the first to showcase the benefits of approximate error correction to quantum communication. They constructed a 4 qubit code which could correct a single amplitude dampening error, violating the ``Hamming bound" for quantum codes based on the Pauli basis. Several results then studied conditions in which approximate quantum error correction is even possible \cite{Barnum2000ReversingQD, SW02, Devetak2005ThePC, Kretschmann2006TheIT, Bny2009ConditionsFT, MN10, MN12, BO10, Hayden2017ApproximateQE}.

In this context, most related to our work are the results of \cite{Crpeau2005ApproximateQE, Leung2006CommunicatingOA, Bergamaschi2022ApproachingTQ}, who studied the adversarial error channel. While it is well known consequence of the Knill-Laflamme conditions \cite{Knill1996TheoryOQ} that no exact quantum code can correct from more than $n/4$ adversarial errors, \cite{Crpeau2005ApproximateQE} constructed approximate quantum codes which could correct $(n-1)/2$ adversarial errors, all the way up to the no-cloning bound! Their codes relied on an insightful combination of quantum error correction, quantum authentication, and classical secret sharing, but came at a cost of an exponentially large alphabet size and assymptotically decaying rate. 

\cite{Leung2006CommunicatingOA, Bergamaschi2022ApproachingTQ} address these issues using the list-decoding of stabilizer codes. \cite{Leung2006CommunicatingOA} presented a quantum code over qubits which beat the quantum Gilbert-Varshamov (GV) bound, albeit relied on an inefficient decoding algorithm, and a secret key shared between sender and receiver. \cite{Bergamaschi2022ApproachingTQ} showed how to efficiently approach the quantum Singleton bound over (large, but) constant-sized alphabets, by combining purity testing codes and classical secret sharing with efficient constructions of list-decodable stabilizer codes. While in this work we share many of the same tools as \cite{Bergamaschi2022ApproachingTQ}, what distinguishes our approaches is that (1) our goal is to design codes over qubits, and (2) we correct from a number of (erasure) errors much larger than the underlying code distance. Thereby, we cannot rely on techniques such as ``code-blocking" to hide shared randomness back into the code\footnote{As similarly done by \cite{Crpeau2005ApproximateQE} and \cite{Hayden2017ApproximateQE} to approach the no-cloning bound.}, nor can we leverage ``privacy" by relying on exact local-indistinguishability. 

\textbf{Quantum Secret Sharing.} \cite{Gottesman1999TheoryOQ, Cleve1999HOWTS, Smith2000QuantumSS, Crpeau2005ApproximateQE} showed that exact and approximate quantum error correction and quantum secret sharing are deeply connected. In a quantum $(n, t)$-threshold secret sharing scheme, a secret quantum state is encoded into $n$ shares such that any $t$ can be used to reconstruct the secret, but any $\leq t-1$ contain no information about the secret. By definition, this is already a quantum error correcting code, and \cite{Cleve1999HOWTS} noted that the converse is also true. Via the information-disturbance paradigm, even approximately correcting from quantum erasures implies the approximate privacy of the erased subset \cite{Kretschmann2006TheIT}.

A natural extension to threshold secret sharing is the concept of a $(n, t, p)$ \textit{ramp} secret sharing scheme, where the ``reconstruction" and ``privacy" thresholds are distinct. Any $t$ suffice to reconstruct the secret, but no $p$ of them reveal any information. Our result in \cref{theorem:main} can be understood as a near-optimal quantum ramp secret sharing scheme encoding secrets of $r\cdot n$ qubits into $n$ binary shares, where $t = (\frac{1+r}{2}+\gamma)\cdot n, p= (\frac{1-r}{2}-\gamma)\cdot n$, with recovery and privacy error $2^{-\Omega(\gamma\cdot n)}$.

\textbf{Non-Malleable Codes.} The first construction of an efficient non-malleable code by \cite{Dziembowski2010} was for ``bit-wise" tampering functions, which is the conceptual classical analog to the qubit-wise channels we study. Their construction was based on combinations of AMD codes and error-correcting codes (with extra secret-sharing properties). A series of works then built on their results by developing explicit and ``minimal redundancy" (rate $1-o(1)$) constructions \cite{Cheraghchi2013NonmalleableCA, nm-permutations, nm-rate-optimizing}. Qualitatively, our construction of tamper detection codes is inspired by \cite{Cheraghchi2013NonmalleableCA}, who combined a ``concatenated" non-malleable code with a pseudorandom permutation. The secret key for the permutation was then itself hidden back into the code, using a smaller, sub-optimal non-malleable code. 

Outside of the bit-wise setting, other notable examples of non-malleable codes include the well studied ``split-state'' model (where the codeword is segmented into $t = O(1)$ parts, instead of $n$ as in the bit-wise setting) \cite{DKO13,adl14, KOS17, ASK22, split-state-review}, as well as AC$^0$ \cite{marshalllowdepth} and poly-sized \cite{marshalllarge} circuits. 

\textbf{In concurrent work} \cite{Boddu2023SplitStateNC} introduced the notion of a non-malleable code for quantum messages, or a \textit{quantum non-malleable code}, and developed applications to quantum secret sharing schemes. Informally, in their definition, the receiver Bob is tasked with either outputting the original quantum message $\psi$ (preserving any side-entanglement), or a fixed state $\sigma$ (decoupled from any side-entanglement). They constructed said codes in the 2-split-state model (which is the hardest), even against entangled adversaries, by combining \textit{quantum-secure} versions of non-malleable codes and extractors \cite{Boddu2022NonMalleableCI, Batra2023QuantumSN} with quantum authentication schemes. While our works share the high level approach of using a non-malleable secret key to encrypt a quantum state, realizing this approach in our two (incomparable) settings requires different ideas and techniques.

We remark that while our (un-entangled) qubit-wise adversarial model is strictly weaker than their (entangled) split-state model, we emphasize that our tamper detection security guarantee is fundamentally stronger than non-malleability - and indeed, tamper detection is impossible against adversaries with unbounded entanglement. We further discuss the differences between these models in \cref{section:overview-keyless}.


\section{Technical Overview}
\label{section:overview}

In this section, we present a brief overview of the main ideas in this work. We assume familiarity with basic definitions of Pauli operators and stabilizer codes. A review of these concepts is presented in \cref{section:prelim}.

In \Cref{section:overview-sparse-pauli-channel}, we show how the intuition that PMDs detect Pauli errors without perturbing the corrupted code-states can be bootstrapped into correcting from ``sparse" Pauli channels. Then, in \Cref{section:overview-qlde}, we overview how list-decoding from erasures can be used to reduce approximate quantum erasure correction to correcting from sparse Pauli channels.

In \Cref{section:overview-keyless}, we overview our construction of tamper detection codes for qubit-wise channels from classical non-malleable codes, PMD codes, and stabilizer codes. Finally, in \Cref{section:overview-explicit}, we overview our explicit constructions of PMD codes from the PTC codes of \cite{barnum2002authentication}.

\subsection{The Sparse Pauli Channel Lemma}
\label{section:overview-sparse-pauli-channel}

The motivating idea behind our use of PMDs in approximate quantum error correction is to ask when a Pauli error $E$ can also be identified (and corrected), if it can be detected. To make this concrete, consider the following distinguishing task:

\begin{task} [Sparse Pauli Channel Correction]
    Let $\ket{\psi}\in \Pi$ be a $\PMD$ code-state and $E_1, E_2$ be two Pauli operators. Then, given the quantum state $E_i\ket{\psi}$ for unknown $i$ and the classical description of $E_1, E_2$, can we recover $\ket{\psi}$?
\end{task}

Naturally, there is a simple algorithm to recover a state close to $\ket{\psi}$. First apply $E_1^\dagger$ and measure $(\Pi, \mathbb{I}-\Pi)$. If the outcome is $\mathbb{I}-\Pi$, then, by the gentle-measurement lemma (see \Cref{claim:pmdapproximation}), we are left with a state $\approx E_1^\dagger E_2 \ket{\psi}$. One can then revert $E_1$, attempt the next correction $E_2^\dagger$, and measure $\Pi$ again. 

We refer to this task as ``Sparse Pauli Channel Correction" because it encapsulates the case where a receiver, Bob, is handed a probabilistic mixture of $\psi\in \Pi$ corrupted by a random Pauli error chosen from a small, known set of errors. In \Cref{claim:singlepaulirecovery} we formalize this approach, and show how to bootstrap the gentle-measurement lemma into recovering $\psi$ even when the error comes from a short list of $L$ possible candidate errors:

\begin{lemma}[Informal, \Cref{claim:singlepaulirecovery}]\label{lemma:results-sparsepauli}
Let $\mathcal{E} = \{E_1, \cdots, E_L\}$ be a list of $L$ distinct $n$ qubit Pauli operators. Then, given a classical description of $\mathcal{E}$, there exists a unitary $\Dec_\mathcal{E}$ on $n+L$ qubits which approximately recovers states in $\Pi$ from errors in $\mathcal{E}$:
\begin{equation}
    \| \Dec_\mathcal{E}\big(E\ket{\psi}\otimes \ket{0^L}\big) - \ket{\psi}\otimes \ket{\Aux_E}\|_2\leq 2\cdot L\cdot \epsilon \text{ for all code-states } \ket{\psi}\in \Pi \text{ and }E\in \mathcal{E},
\end{equation}
for some $L$ qubit state $\ket{\Aux_E}$ which only depends on $E$ and $\mathcal{E}$.
\end{lemma}

If $L = o(1/\epsilon)$, then PMDs can correct from Sparse Pauli Channels with asymptotically decaying error. In fact, in \Cref{section:AQEC} we present a quantum algorithm which can coherently correct from superpositions of errors in the span of $\mathcal{E}$, and thus generalizes the proposition above on mixtures of Pauli errors. As we later discuss, this coherent correction step is what later allows us to recover from \textit{adaptive} erasure errors.

\subsection{From List Decoding to Approximately Correcting Quantum Erasures}
\label{section:overview-qlde}

Our quantum erasure code construction combines PMD codes with stabilizer codes which are \textit{list-decodable from erasures}, a variant of the list-decodable stabilizer codes studied by \cite{Leung2006CommunicatingOA, Bergamaschi2022ApproachingTQ}. In particular, we consider their code composition: 
\begin{equation}
  \Enc(\ket{m})\equiv   \Enc  \ket{m}\otimes \ket{0^{a_\PMD}}\otimes \ket{0^{a_Q}} = \Enc_Q \big(\Enc_{\PMD} \otimes \mathbb{I}^{a_Q}\big) \ket{m}\otimes \ket{0^{a_\PMD}}\otimes \ket{0^{a_Q}}
\end{equation}

where the unitaries $\Enc_Q, \Enc_{\PMD}$ encode into the erasure list-decodable code and the PMD code respectively, and $a_\PMD, a_Q$ correspond to ancilla qubits. \Cref{fig:setup} represents this encoding setup, and the erasure process.

\begin{figure}[ht]
    \centering
    \includegraphics[width=0.7\textwidth]{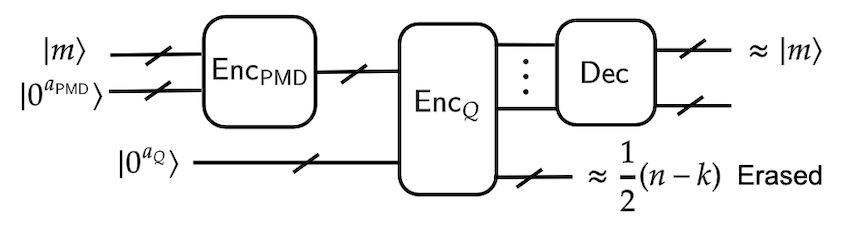}
    \caption{The composition of a $\PMD$ and QLDE codes, and an erasure error.}
    \label{fig:setup}
\end{figure}

One of our main technical contributions is designing an efficient decoding channel $\Dec$ for our code construction, that is, a reduction from approximate erasure correction to list-decoding from erasures and Pauli Manipulation Detection. In the rest of this subsection, we overview how an outer list-decoding step can be used to reduce the erasure channel to a sparse Pauli channel, which we can then handle with the inner PMD code and \Cref{lemma:results-sparsepauli}.

In classical coding theory, list-decoding \cite{Elias1957ListDF, woz1958} arose as a natural relaxation to ``unique'' decoding, where the decoder is tasked with recovering a list of candidate messages from a corrupted codeword, when recovering a unique one may no longer be possible. Formally, a code $C\subset \Sigma^n$ is said to be $(\tau, L)$ list-decodable if there are at most $L$ codewords of $C$ in any hamming ball of radius $\tau\cdot n$ in $\Sigma^n$. If $C$ is linear, with parity check matrix $H$, then this definition implies that at most $L$ additive errors $e_1, \cdots, e_L$ of bounded weight $\leq \tau\cdot n$ can have the same parity check syndrome $s = He_i$.

Analogously to linear codes, in a stabilizer code which is $(\tau, L)$ list-decodable, there are at most $L$ logically-distinct Pauli errors of weight $\leq \tau\cdot n$ consistent with any syndrome measurement \cite{Leung2006CommunicatingOA, Bergamaschi2022ApproachingTQ}. Operationally, what this implies is that measuring the syndrome of a bounded weight adversarial error channel, collapses the channel into a mixture over Pauli errors in a discrete list $\mathcal{E}=\{E_1, \cdots, E_L\}$, i.e., a sparse Pauli channel!

The notion of list-decoding \textit{from erasures} was introduced by \cite{Guruswami2001ListDF}, and conceptually corresponds to the natural extension to list-decoding when the error locations are known. Based on their work, we define:

\begin{definition}\label{definition:overiew-QLDE}
A stabilizer code on $n$ qubits is $(\tau, L)$ list-decodable from erasures (QLDE) if for all syndrome vectors $s$ and subsets $T\subset [n]$, $|T|\leq \tau\cdot n$, there are at most $L$ logically-distinct Pauli operators supported on $ T$ and of syndrome $s$.
\end{definition}

In \Cref{section:QLDE}, we show that a syndrome measurement on a QLDE code analogously ``collapses" the quantum erasure channel into a mixture of Pauli errors in a list $\mathcal{E}$. More importantly, (given the syndrome) we reason that one can always compute such a list of Pauli errors $\mathcal{E}=\{E_1, \cdots, E_L\}$ efficiently in terms of $n$ and $L$, since list-decoding stabilizer codes \textit{from erasures} actually corresponds to solving a linear system.\footnote{In contrast to normal list-decoding of linear codes, erasure list-decoding linear codes simply corresponds to solving a linear system, and thereby is efficient in terms of $n$ and $L$. See \cite{Guruswami2001ListDF}, or \Cref{section:QLDE} for a review.}

Together with our Sparse Pauli Channel \Cref{lemma:results-sparsepauli}, in \Cref{section:AQEC} we conclude that the code composition $\Enc = \Enc_Q\circ \Enc_\PMD$ efficiently corrects erasures:

\begin{lemma}[Informal, \Cref{lemma:aqeccsfrompmds}]\label{lemma:results-aqeccsfrompmds}
    The code composition of an $\epsilon$-$\PMD$ with a $(\tau, L)$ erasure list-decodable code corrects from a $\tau$ fraction of adversarially chosen erasures with up to a $O(\epsilon^{1/2}\cdot L^{3/4})$ recovery error in poly$(n, L)$ time. 
\end{lemma}

To instantiate \Cref{lemma:results-aqeccsfrompmds}, in \Cref{section:QLDE} we show how to build $(\tau, L)$-QLDE codes from the CSS construction \cite{Calderbank1996GoodQE, Steane1996SimpleQE} of classical linear codes which are list-decodable from erasures, and present a simple randomized code construction. Our randomized construction is based on random CSS codes, and we show that they are list-decodable from a number of erasures approaching the quantum Singleton bound via the classical results of \cite{Guruswami2001ListDF, Ding2014ErasureLC}:

\begin{theorem}[Informal version of \Cref{corollary:randomcss}] \label{theorem:results-randomcss}
For all $0<r, \gamma <1$ and sufficiently large positive integer $n$, a random rate $r$ CSS code on $n$ qubits is $(\frac{1}{2}(1-r-\gamma), 2^{O(1/\gamma)})$-QLDE with probability $\geq 1-n^{O(1)}\cdot 2^{-\Omega((1-r) \cdot n)}$.
\end{theorem}

Unfortunately, we don't know yet of deterministically constructable stabilizer codes with these list-decoding guarantees (see \Cref{section:discussion}). However, our randomized constructions are efficiently sampleable, encodable and erasure list-decodable, and their construction succeeds (i.e. has the desired list decoding property) with exponentially high probability.

\subsection{Quantum Tamper Detection for Qubit-Wise Channels}
\label{section:overview-keyless}

The definition of quantum tamper detection we consider\footnote{In the first posting of this work, we stated an alternative definition referred to as ``keyless authentication" based on \cite{Hayden2016TheUC}. For this version, we changed the definition to the more modern simulation paradigm originally set by \cite{Boddu2023SplitStateNC} in the context of quantum non-malleable codes. However, none of our proofs were modified. } is a special case of that of \cite{Boddu2021TamperDA}, and can be understood as a keyless version of the definition of quantum authentication by \cite{Hayden2016TheUC}. We consider a public and keyless pair of channels $(\Enc, \Dec)$, where $\Enc$ encodes $k$ qubits into an $n$ qubit mixed state, and $\Dec$ decodes $n$ qubits into $k+1$ qubits, consisting of a message register $M$ and a flag $F$ indicating acceptance $\ket{\Acc}$ or rejection $\ket{\Rej}$. 

\cref{definition:qa} quantifies standard correctness and robust error-detection guarantees, when a quantum state on the register $M$ (possibly entangled with a side register $R$) is encoded into the code. 

\begin{definition} \label{definition:qa}
    A pair of quantum channels $(\Enc, \Dec)$ is an $\epsilon$-secure quantum tamper detection code for qubit-wise channels if they satisfy
\begin{enumerate}
    \item (Correctness) In the absence of tampering, for all messages $\psi_{MR}$: $( \Dec\circ \Enc\otimes \mathbb{I}_R )(\psi_{MR}) = \Acc\otimes \psi_{MR}$.
    \item (Tamper Detection) For every set of $n$ single qubit channels $(\Lambda_1, \cdots, \Lambda_n)$, there exists a constant $p_\Lambda\in [0, 1]$ such that, 
    \begin{equation}
        \forall \psi_{MR}: \bigg( \Dec\circ \big(\otimes_i^n \Lambda_i \big)\circ  \Enc\otimes \mathbb{I}_R \bigg)(\psi_{MR}) \approx_\epsilon p_\Lambda\cdot \Acc\otimes \psi  + (1-p_\Lambda)\cdot  \bot\otimes \psi_R  
    \end{equation}

    Where the error is measured in trace distance, and $\bot$ on $M$ and $F$ indicates the message is rejected. 
\end{enumerate}
\end{definition}

This definition allows the adversaries to share randomness, but explicitly does not allow them to communicate, share entanglement, nor be entangled with the message. Therefore, it is not composable in the standard (authentication) sense. Nevertheless, it still provides meaningful guarantees in a ``one-shot" setting. For instance, it can be used to establish entanglement, a secret key, or teleport over a state between Alice and Bob.\footnote{Here, we comment on the distinction between this model and the recent quantum non-malleable codes of \cite{Boddu2023SplitStateNC}. In a tamper-detection code, if Bob accepts with non-negligible probability $\omega(\epsilon)$, then he is guaranteed to have preserved entanglement with Alice. In a quantum non-malleable code, the adversaries may have decoupled Alice and Bob, and Bob will be none-the-wiser.}

In the rest of this subsection we overview \Cref{theorem:auth-results}. We begin by presenting a simplification of our main result, a tamper detection code of rate approaching $1/3$, which showcases the main techniques. For completeness, we present a self-contained proof of this simpler construction in \Cref{section:auth}. Afterwards, we present a high-level description of how we leverage code concatenation and pseudorandomness to achieve our rate $1-o(1)$ construction.

\subsubsection{The Rate $1/3$ Construction}

The key tool we use is a classical non-malleable code against bit-wise tampering functions \cite{Dziembowski2010, Cheraghchi2013NonmalleableCA, nm-permutations, nm-rate-optimizing}. Informally, a classical non-malleable code $(\Enc_\NM, \Dec_\NM)$ is a relaxation of error-correction and error-detection codes, where the decoder is tasked with either outputting (1) the original message, (2) rejecting $\bot$, \textbf{or} (3) outputting a message which is completely uncorrelated from the original message. The role of the non-malleable code in our construction is to establish a ``non-malleable secret key" between Alice and Bob. That is, a $\approx 2/3$ fraction of the qubit-wise channels will be used to send over an encoding $\Enc_{\NM}(s)$ of a uniformly random bitstring $s$, which will later be used to encrypt Alice's message state. 
 
 Alice encodes her message $\ket{\psi}$ first into the composition $\Tilde{Q}$ of a PMD code, and a stabilizer code $Q$ of block-length $N_Q$ and near-linear distance $d$. Then, she one-time-pads her state with a random $N_Q$ qubit Pauli $P^s$, indexed by a $2\cdot N_Q$ bit long secret key $s$, which is then encoded into $\NM$ (see \Cref{fig:third}). 
\begin{gather}
    \Enc(\psi) = \mathbb{E}_s \Enc_{\NM}(s)\otimes  P^s\cdot \Enc_{\Tilde{Q}}(\psi)\cdot (P^s)^\dagger, \text{ where } \Enc_{\Tilde{Q}}= \Enc_Q\circ\Enc_{\PMD}
\end{gather}
If the $\PMD$ code, the stabilizer code $Q$, and the classical non-malleable code $\NM$ all have rate $\approx 1 - o(1)$, then the resulting construction has rate $\approx \frac{1}{3}(1-o(1))$.

Our decoding algorithm is relatively simple to describe. After receiving Alice's message, Bob first decodes the non-malleable code using $\Dec_\NM$. Assuming he doesn't immediately reject, Bob obtains some string $\Tilde{s}$, and uses it to revert the Pauli one-time-pad applied on the quantum half of the encoding. After doing so, Bob measures the syndrome of $Q$, and if successful, projects onto the PMD code-space. Assuming all of these steps accept, at the end of the protocol Bob outputs the ``message" register of the PMD. 

\begin{figure}[h]
    \centering
        \includegraphics[width=.6\textwidth]{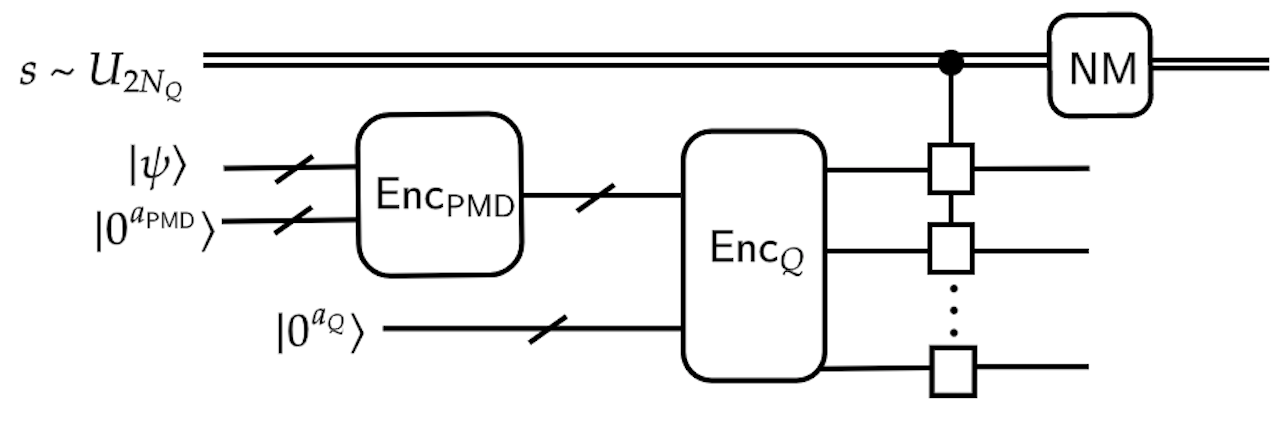} 
        \caption{$\Enc$ for the rate $1/3-o(1)$ construction. The unlabeled controlled gate represents the Pauli One-Time-Pad.}
        \label{fig:third}
\end{figure}

Informally, the classical non-malleability guarantees that the distribution over Bob's recovered key $\Tilde{s}$ is close to a convex combination over the original key $s$, and an uncorrelated key drawn from some arbitrary distribution $\mathcal{D}^\Lambda$ (\Cref{claim:classical-nm-third}). In our analysis, we can essentially consider these two cases separately:

In the first case, if Bob recovers the original key $s$, then on average over random $s$ he receives the outcome of a Pauli channel:
\begin{equation}
    \mathbb{E}_s (P^s)^\dagger \circ \Lambda\circ P^s\circ \Enc_{\Tilde{Q}}(\psi) = \sum_{E\in \text{Pauli}} |c_E|^2 \cdot E\cdot \Enc_{\Tilde{Q}}(\psi)\cdot E^\dagger
\end{equation}
Indeed, the operation on the left-hand-side above is the well known Pauli Twirl operation \cite{Aharonov2008InteractivePF} (\cref{fact:twirl}). Since the effective channel is a Pauli channel, we can rely on the properties of the PMD code to detect any adversarial tampering (\Cref{claim:third-auth-recovered}).

Conversely, if Bob does not recover the original key $s$, then since Pauli's are 1-Designs (\cref{fact:1design}), on average over a random Pauli he receives a uniformly random (encrypted) message:
\begin{equation}
    \mathbb{E}_s \Lambda\circ P^s \circ \Enc_{\Tilde{Q}}(\psi_{MR}) = \Lambda\circ \bigg(\mathbb{E}_s P^s \cdot \Enc_{\Tilde{Q}}(\psi_{MR}) \cdot (P^s )^\dagger\bigg) = \psi_R\otimes \Lambda(\mathbb{I}/2^{N_Q}) = \psi_R\otimes_i^{N_Q} \Lambda_i(\mathbb{I}/2)
\end{equation}

Moreover, from Bob's perspective, the state he receives $\otimes_i \Lambda_i(\mathbb{I}/2)$ is completely uncorrelated from Alice's message $\psi$! Conceptually, this guarantees that our code construction is already sort of a quantum ``non-malleable code". However, we can in fact say something much stronger. 

What makes this construction a tamper detection code, instead of just a non-malleable code, is the fact that if Bob does not receive the original key then his syndrome measurement will reject with high probability. This follows from the fact that his received state, on average over random $s$, is the (unentangled) \textit{product state} $\otimes_i \Lambda_i(\mathbb{I}/2)$. However, product states should be very far from the highly entangled code-spaces of stabilizer codes. We leverage a lemma by Anshu and Nirkhe \cite{Anshu2020CircuitLB} (\cref{lemma:clb}) to show that the probability Bob's syndrome measurement accepts (that is, outputs syndrome 0), if the key has been tampered with is $\leq 2^{-\Omega(d^2/N_Q)}$, exponentially decaying with the distance $d$ of the stabilizer code (\Cref{claim:third-auth-not-recovered}). 

Here, it is instructive to pause briefly and discuss why this or related constructions wouldn't be able to achieve tamper-detection if the adversaries were entangled. The first obstacle is that the classical non-malleable codes we used are, apriori, not robust to adversaries which can leverage quantum correlations. Nevertheless (modulo rate considerations), this obstacle is surmountable by leveraging recent constructions of quantum-secure non-malleable codes and extractors \cite{Boddu2022NonMalleableCI}, see recent work by \cite{Boddu2023SplitStateNC}. The truly fundamental issue is to quantify how far the adversaries pre-shared state is from the code-space of the stabilizer code. Indeed, if they could pre-prepare an arbitrary entangled state, they could replace the code-state themselves, ensuring that tamper-detection is impossible. 

\subsubsection{Rate Amplification}

Our approach to improve the rate is inspired by a construction of classical non-malleable codes by \cite{Cheraghchi2013NonmalleableCA}. The bottleneck in the previous construction is the amount of classical randomness used, and thus our first step is to replace the uniformly random Pauli one-time-pad $P^s$ by a $t$-wise independent Pauli Pad $P^{G(s)}$, generated using a short seed of length $\sigma = O(t\log k)$. Unfortunately, under a $t$-wise independent pad, the quantum half of $\Enc$ is no longer perfectly encrypted, and our proof techniques in the previous section (the Pauli Twirl and \cite{Anshu2020CircuitLB}'s result) break.

To address these issues, we use quantum code concatenation. Alice first encodes her message $\psi$ into a high rate, near-linear distance stabilizer code, and then encodes the symbols of the outer code\footnote{Technically, the symbols are first bundled up into blocks of size $k/\log^c k$, and then each block is encoded into $\Tilde{Q}$.} into the composition $\Tilde{Q}$ of a PMD code and another high rate stabilizer code of smaller block-length (See \Cref{fig:concat}).
\begin{figure}[h]
    \centering
        \includegraphics[width=.45\textwidth]{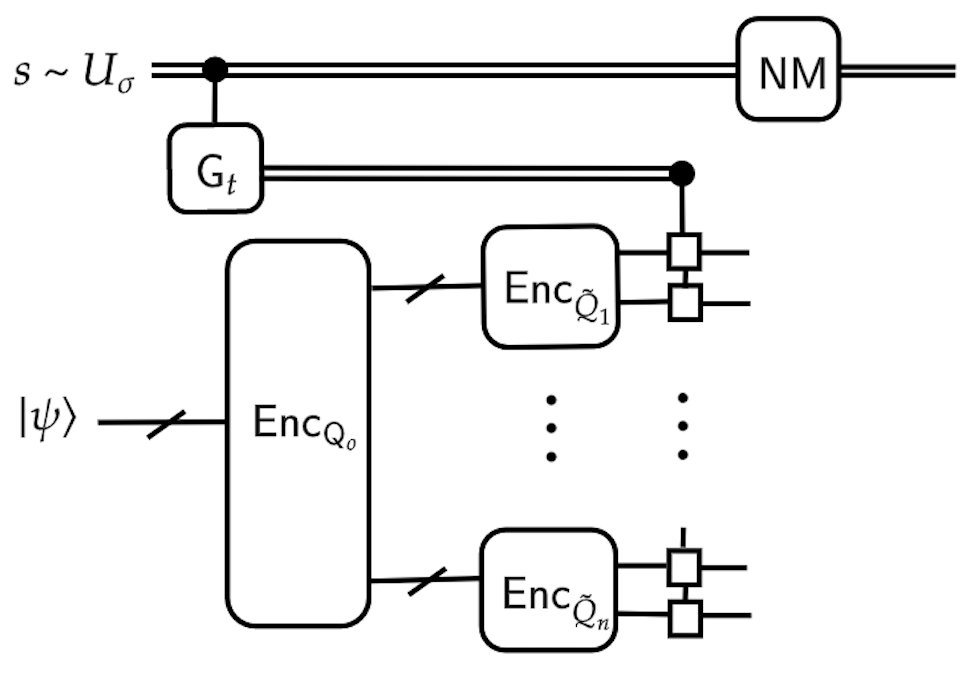} 
        \caption{The concatenated code with a pseudorandom Pauli Pad. Ancilla qubits removed for clarity.}
        \label{fig:concat}
\end{figure}

 Intuitively, if the block-length of the inner codes is $b < t$, then via the $t$-wise independence the marginal density matrix on each inner code is truly encrypted, and we can leverage the proof techniques of the uniformly random case. While this intuition does capture the case where Bob does not recover the key, the case where Bob recovers the original key is a little more subtle. Since $G(s)$ is only $t$-wise independent, we can no longer argue that the effective channel on the entire code is simply a Pauli channel.

The key insight to address the case where the key is recovered is to further divide into cases on whether each block is significantly tampered, or whether each tampering channel $\Lambda_j$ is ``sparse" in the Pauli basis (See \Cref{definition:eta-pauli}). While we defer the details to \Cref{section:analysis-rate1}, the underlying intuition is that if many of the channels in any given inner block $\Tilde{Q}_i$ are not ``sparse" in the Pauli basis, then the Pauli twirl argument should still imply that block $i$ rejects with high probability (\Cref{claim:not-many-eta}). Conversely, if many of the channels in any given inner block are ``sparse", then the PMD in that block should detect the presence of errors without distorting the state (\Cref{claim:many_eta}). 

\subsection{Explicit PMDs from Pairwise-Detectable Purity Testing}
\label{section:overview-explicit}

To conclude our technical overview, we discuss our construction of PMD codes from the stabilizer Purity Testing Codes (PTCs) by \cite{barnum2002authentication}. A PTC is a set of stabilizer codes which detects every Pauli error ``on average":

\begin{definition} 
    A set of stabilizer codes $\{Q_k\}_{k\in K}$ is called an $\epsilon$-strong $\PTC$ if for all non-trivial Pauli operators $E\neq \mathbb{I}$, it holds that
    \begin{equation}
        \mathbb{P}_{k\leftarrow K}\big[ E\in N_k\big] \leq \epsilon.
    \end{equation}
    where $N_k$ is the normalizer group (the undetectable errors) of $Q_k$.
\end{definition}

In other words, the probability that $E$ is detectable is $\geq 1-\epsilon$, for a random choice of code $k\leftarrow K$. \cite{barnum2002authentication} presented a particularly key-efficient family of PTCs, and used them to develop quantum message authentication codes:

\begin{theorem}[\cite{barnum2002authentication}]
For every sufficiently large $n, \lambda\in \mathbb{N}, \lambda|n$, there exists a set of  $[[n, n-\lambda]]_2$ stabilizer codes $\{Q_\alpha\}_{\alpha\in \mathbb{F}_{2^\lambda}}$ which form an $\epsilon$-strong $\PTC$ with error $\epsilon \leq n\cdot 2^{-\lambda}$.
\end{theorem}

A natural ``candidate" $\PMD$ is simply a superposition over keys and PTC encodings. That is, if each $Q_k$ has a unitary encoding map $\Enc_k$ acting on some number $\Anc_{\PTC}$ of ancillas and the message state $\ket{\phi}$, then our PMDs encode $\ket{\phi}$ as a uniform superposition over keys (on a ``key" register) and PTC encodings (on a ``code" register):
\begin{equation}
    \Enc\ket{\phi}  = |K|^{-1/2}\sum_k \ket{k} \otimes \Enc_k \ket{\phi}
\end{equation}
There is a natural associated projection $\Pi = \Enc (\mathbb{I}\otimes \ketbra{0}^{\Anc_{key}+\Anc_{\PTC}})\Enc^{\dagger}$ onto the codespace as well (See \Cref{section:PMDs}). 

Intuitively, one expects that Pauli errors supported only on the ``code" register should be handled by the PTC property. The challenge really lies in what happens when the key register itself is corrupted. To build a rough idea on our proof technique, consider the following ``Key Manipulation" experiment: 

    Let us encode some state $\ket{m}$ into the $a$th stabilizer code $Q_a$, and then measure the syndrome of the $b$th code. When $a=b$, there is no tampering and the syndrome measurement naturally ``accepts" (i.e. we measure syndrome 0). When $a\neq b$, the probability we mistakenly accept is
    \begin{equation}
        \bra{m} \Enc^\dagger_a \Pi_{b}\Enc_a\ket{m} \leq \| \Pi_a \Pi_{b}\Pi_a \|_\infty,
    \end{equation}

    where $\Pi_k$ corresponds to the projection onto the $k$th stabilizer code-space $Q_k$. If $Q_k$ has stabilizer group $S_k$, then the projector onto $Q_k$ can be written as $\Pi_k = |S_k|^{-1} \sum_{\sigma\in S_k}\sigma$. We note that for any Pauli $\sigma$, $\Pi_k \sigma \Pi_k = 0$ unless $\sigma\in N_k$, thus

\begin{equation}
   \| \Pi_a \Pi_{b}\Pi_a \|_\infty \leq |S_b|^{-1} \sum_{\sigma\in S_{b}\cap N_a} \|\sigma\|_\infty \leq |S_b|^{-1}\cdot |S_{b} \cap N_a|
\end{equation}

That is, a measure of how many stabilizers of $S_b$ are ``undetectable'' to the code $Q_a$. We remark that on average over $a, b$, and, assuming all the stabilizer groups have size $|S_b| = 2^\lambda$, linearity of expectation and the $\epsilon$-strong condition tells us that 
\begin{equation}
   \mathbb{E}_{a\leftarrow K} \| \Pi_a \Pi_{b}\Pi_a \|_\infty \leq 2^{-\lambda}\bigg(1 + \sum_{E\in S_b\setminus\{\mathbb{I}\}}\mathbb{P}_{a\leftarrow K}[E\in N_a]\bigg)  \leq 2^{-\lambda} + \epsilon
\end{equation}
That is, these code-space projectors are almost orthogonal on average (for random $a, b$).

Unfortunately, this doesn't tell us much about correlated pairs of keys $(a, b)$. For instance, what about the pairs $(a,b=a+s)$, where $s$ is some fixed ``bit-flip" error? In particular, let us assume the set of keys $K$ form a group under addition, and consider pairs $(a, a+s)$ for uniformly random $a\leftarrow K$ and a fixed non-zero shift $s\in K$. We refer to the set of stabilizer codes $\{Q_a\}_{a\in K}$ as pairwise detectable, if over random $a\in K$ the probability some operator of $S_a$ is undetectible to the code $Q_{a+s}$ is small.

\begin{definition}
    A set of stabilizer codes $\{Q_k\}_{k\in K}$ is called $\delta$ pairwise-detectable if, for all $s\neq 0$, 
    \begin{equation}
        \mathbb{P}_{k\leftarrow K}\big[S_k \cap N_{k+s} \neq \{\mathbb{I}\}\big] \leq \delta
    \end{equation}
    
\end{definition}

In which case, $\mathbb{E}_{a\leftarrow K} \| \Pi_a \Pi_{a+s}\Pi_a \|_\infty \leq 1/2^\lambda + \delta$ for all non-zero $s\in K$. In \Cref{section:PMDs}, we prove that pairwise-detectability is precisely what we need to construct PMD codes, resulting in the following lemma:

\begin{lemma}[\Cref{lemma:pmds-from-ptcs}, restatement]
    If a set of $[[n, n-\lambda]]_2$ stabilizer codes $\{Q_k\}_{k\in K}$ is a $\epsilon$-strong PTC and is $\delta$ pairwise-detectable, then the code defined by $\Enc, \Pi$ above is a $(n+\log K, n-\lambda, \Delta)$-$\PMD$ with $\Delta\leq  \max(\epsilon, \sqrt{2^{-\lambda} + \delta})$
\end{lemma}

It only remains to instantiate our lemma with a viable set of stabilizer codes. In \Cref{subsection:PTCs}, we prove that the PTCs introduced by \cite{barnum2002authentication} provide a particularly key-efficient family of pairwise detectable stabilizer codes, which gives us our PMD construction of \Cref{theorem:results-pmdexplicit}.

\begin{claim}[\Cref{claim:detectibleproof}, restatement] \label{claim:pairwisebarnum}
    The family of $[[n, n-\lambda]]_2$ stabilizer codes considered by \cite{barnum2002authentication} is $\leq n \cdot 2^{1-\lambda}$ pairwise-detectable. 
\end{claim}

\subsection{Organization}
We organize the rest of this paper as follows. 

\begin{enumerate}
    \item [--] In \Cref{section:discussion} we discuss and raise interesting questions left by our work. 
    \item [--] Before diving into the proofs, we begin in \Cref{section:prelim} by presenting a brief background on quantum information theory, as well as the relevant definitions and statements on quantum error-correcting codes and non-malleable codes.
    \item [--] In \Cref{section:PMDs} we discuss properties of PMDs, present their construction from pairwise-detectable purity testing codes, and show that the purity testing codes of \cite{barnum2002authentication} are pairwise-detectable.
    \item [--] In \Cref{section:QLDE} we introduce and construct quantum codes which are list decodable from erasures.
    \item [--] In \Cref{section:AQEC}, we show how to construct approximate quantum codes for erasures/quantum secret sharing schemes from PMDs and erasure list-decoding. 
    \item [--] Finally, in \Cref{section:auth}, we present our constructions of tamper detection codes for qubit-wise channels.
\end{enumerate}

\section{Discussion}
\label{section:discussion}
We dedicate this section to a discussion on the interesting directions and open problems raised here.

The first question lies in developing new explicit constructions and lower bounds for PMD codes. While lower bounds and near-optimal constructions have been studied in a multitude of settings for AMD codes  \cite{Cramer2008DetectionOA, Cramer2015OptimalAM, Bao2016NewEA, Huczynska2016ExistenceAN, Huczynska2018WeightedED}, our attempts at adapting their techniques to the quantum setting have been far from fruitful.  

In recent work, \cite{BenAroya2020NearoptimalEL} devised (classical) binary codes which are near-optimally erasure list-decodable. That is, they constructed codes which could correct from $(1-\tau)$ fraction of erasures with list sizes $L = O(\log 1/\tau)$ and rate $\tau^{1+\gamma}$, for any value of $\tau > n^{-\Omega(1)}$. These codes have the remarkable property that they provably have smaller list sizes (for fixed $\tau$) than the best possible linear code \cite{Guruswami2001ListDF}. Since stabilizer codes are built on classical linear codes (and thus inherit their lower bounds), one could ask whether efficient non-stabilizer codes exist with even better list sizes than our random CSS construction. Unfortunately, even asking for deterministically constructible erasure list-decodable stabilizer codes seems to be a challenge. 

An interesting case of our construction corresponds to the $k=1$ case, that is, encoding a single logical qubit. In this limit, our codes encode the logical qubit into $n$ physical qubits such that it is recoverable from a $\frac{1}{2} - O(\log^{-1} n)$ fraction of erasures. It is the list sizes in our construction that arise as the key bottleneck to further approaching the no-cloning bound of $1/2$. One could analogously ask, is there an efficient quantum code to approximately correct from a $\frac{1}{2} - n^{-\Omega(1)}$ fraction of erasures?

Finally, the main and arguably most interesting question is, what other adversarial quantum channels can be used to send authenticated quantum messages without a secret key? What about non-malleable quantum messages \cite{Boddu2023SplitStateNC}? While we focus on qubit-wise channels as a proof-of-concept here, generalizing our approach to other settings such as LOCC (local operations and classical communication) channels, qubit-wise channels with entangled parties, or the split-state model \cite{split-state-review}, all require fundamentally new insights.\footnote{See \cite{Boddu2023SplitStateNC, Batra2023QuantumSN, Bergamaschi2023OnSQ} for recent progress in this direction.}

\section*{Acknowledgements}

The author would like to thank Sam Gunn, Louis Golowich, Fermi Ma, and Venkat Guruswami for fruitful discussions on PMD codes and quantum list-decoding, Marshall Ball and Nathan Ju for conversations on models of classical and quantum non-malleability, and Yunchao Liu, Zeph Landau and Umesh Vazirani for help on initial versions of the manuscript and presentation.

Finally, special thanks to James Bartusek for suggesting the application to non-malleable codes for qubit-wise channels, leading to our result on quantum authentication.

TB acknowledges support by the National Science Foundation Graduate Research Fellowship under Grant No. DGE 2146752.

\bibliographystyle{alphaurl}
\bibliography{references}

\appendixpage
\appendix

\section{Preliminaries}
\label{section:prelim}

\subsection{Background and Notation}

\subsubsection{Quantum States, Channels, and Distances}

Let $ \mathcal{L}(\mathcal{H})$ denote the set of linear operators on a finite dimensional Hilbert space $\mathcal{H}$. A quantum state $\rho\in \mathcal{L}(\mathcal{H})$ is a linear operator which is positive semi-definite and trace 1, and is said to be pure if $\rho$ is rank $1$. To quantify notions of distances between quantum states, mixed states and channels, we use the Schatten matrix norms. For integer $p$, and any $M\in \mathcal{L}(\mathcal{H})$, $\|M\|_p$ is the vector $L_p$ norm of the singular values of $M$. Of particular attention is the trace distance, which corresponds to the Schatten $1$ norm $\|M\|_1 = \text{Tr}\sqrt{M^\dagger M}$; and the operator (or ``infinity") norm $\|M\|_\infty = \max_{\ket{\psi}} |\bra{\psi} M^\dagger M\ket{\psi}|^{1/2}$.

A quantum channel $\mathcal{N}:\mathcal{L}(\mathcal{H}_1)\rightarrow \mathcal{L}(\mathcal{H}_2)$ is a completely positive and trace preserving (CPTP) linear map. Any quantum channel can be expressed in an operator-sum or ``Krauss" decomposition, defined by sets of linear operators $\{K_i:\mathcal{H}_1\rightarrow \mathcal{H}_2\}$, where $ \mathcal{N}(\rho) = \sum_i K_i \rho K_i^\dagger \text{ and }\sum_i K_i^\dagger K_i = \mathbb{I}_{\mathcal{H}_1}$. To quantify the recovery error in our approximate quantum codes, we use a strong notion of channel distinguishability, the diamond norm distance. The diamond norm is a measure of the channel distinguishability with possibly entangled inputs:

\begin{definition}\label{definition:diamond}
    The diamond norm distance between two quantum channels $\mathcal{N}, \mathcal{M}$ is :
    \begin{equation}
       \|\mathcal{N}-\mathcal{M}\|_\diamond = \max_n\sup_{\rho} \big\|(\mathcal{N}\otimes \mathbb{I}_n)(\rho)-(\mathcal{M}\otimes \mathbb{I}_n)(\rho)\big\|_1
    \end{equation}
\end{definition}

\subsubsection{Finite Fields}
Let $q =  p^m$ be a power of a prime $p$, and denote by $\mathbb{F}_{q}$ to be the Galois field of $p^m$ elements. We refer to the $\mathbb{F}_p$ functional $\text{tr}_{\mathbb{F}_{q}/\mathbb{F}_{p}}:\mathbb{F}_{q}\rightarrow \mathbb{F}_p$ as the trace function, where $\text{tr}_{\mathbb{F}_{q}/\mathbb{F}_{p}}(a) = \sum_{i=0}^{m-1}a^{p^i}$. If $\alpha_1, \cdots, \alpha_m$ is a basis of $\mathbb{F}_q$ over $\mathbb{F}_p$, then one can express $a\in \mathbb{F}_{q} $ as $a = \sum_{i=1}^m a_i \alpha_i$ for $a_i\in \mathbb{F}_p$. We refer to a pair of bases $\alpha = \alpha_1, \cdots, \alpha_m$, $\beta = \beta_1\cdots \beta_m$ of $\mathbb{F}_q$ as dual bases if $\text{tr}_{\mathbb{F}_{q}/\mathbb{F}_{p}}(\alpha_i\beta_j) = \delta_{i,j}$. If $a, b\in \mathbb{F}_{q}$ is expressed as $(a_1, \cdots, a_m)$, $(b_1, \cdots, b_m)$ in the dual bases $\alpha, \beta$ respectively, then the inner product over the basis representation becomes the trace:
\begin{equation}
   \big\langle a,  b\big\rangle =\sum_{i = 1}^{m} a_i b_i  = \sum_{i, j = 1}^{m} a_i b_j \text{tr}_{\mathbb{F}_{q}/\mathbb{F}_{p}}(\alpha_i\beta_j) = \text{tr}_{\mathbb{F}_{q}/\mathbb{F}_{p}}(ab)
\end{equation}

\subsubsection{$q$-ary Pauli Operators}

Let $\omega = e^{2\pi i/p}$. We define $T, R$ to be the `shift' and `phase' operators on $\mathbb{C}^p$,
\begin{equation}
    T  = \sum_{x\in \mathbb{F}_p}\ket{x}\bra{x+1} \text{ and }R =  \sum_{x\in \mathbb{F}_p} \omega^x \ket{x}\bra{x}.
\end{equation}
The operators $T^iR^j$, $i, j\in \mathbb{F}_p$ form the Weyl-Heisenberg operators, an orthonormal basis of operators over $\mathbb{C}^p$. If $a, b\in \mathbb{F}_q$, with representations $(a_1, \cdots a_m)$, $(b_1, \cdots b_m)$ in the dual bases $\alpha, \beta$ respectively, then one can define an o.n. basis of operators over $\mathbb{C}^{q}$: 

\begin{equation}
    E_{a, b} = X^a Z^b  = \bigotimes_{i\in [m]} T^{a_i} R^{b_i} \text{ and thus }  E_{a, b} E_{a', b'} = \omega^{\langle a, b'\rangle -\langle a', b\rangle   } E_{a', b'}E_{a, b}.
\end{equation}

Where $\otimes$ denotes the tensor product. Finally, if $\textbf{a} = (a^{(1)}, a^{(2)}\cdots a^{(n)}),  \textbf{b} = (b^{(1)}, b^{(2)}\cdots b^{(n)}) \in \mathbb{F}_{q}^n$, then one can define operators acting on $\mathbb{C}^{q^n}$ via $E_{\textbf{a}, \textbf{b}} = \otimes_{j\in [n]} E_{a^{(j)}, b^{(j)}}$. The set of $n$ qudit Pauli operators $\mathbb{P}_{q}^n$ consists of the collection of operators $E_{\textbf{a}, \textbf{b}}$, and the $n$ qudit Pauli group $\mathcal{P}_{q}^n$ is the group generated by the $E_{\textbf{a}, \textbf{b}}$ and the phase $\omega \cdot \mathbb{I}_{q^n\times q^n}$. The weight $wt(E_{\textbf{a}, \textbf{b}})$ is the number of locations $j\in [n]$ where either $a^{(j)}, b^{(j)}$ are non-zero, and $\mathbb{P}_{q, \delta}^n\subset \mathbb{P}_{q}^n$ is the set of operators in the group of weight less than $\delta\cdot n$. 

When the underlying Hilbert space $\mathcal{H}$ is otherwise implicit (and has prime power dimension), we denote $\mathbb{P}_{\mathcal{H}}$ as the set of Pauli operators on $\mathcal{H}$. Note that $|\mathbb{P}_{\mathcal{H}}| = \text{dim}(\mathcal{H})^2$.

\subsubsection{Properties of Averages over the Pauli Group}

We require two basic facts on averages over Pauli operators on $\mathcal{H}$.  

\begin{fact}
    [Pauli's are 1 Designs]\label{fact:1design} Let $\rho\in \mathcal{L}(\mathcal{H}_1\otimes \mathcal{H}_2)$ be a quantum state. Then, 
    \begin{equation}
       |\mathbb{P}_{\mathcal{H}_1}|^{-1} \sum_{E\in \mathbb{P}_{\mathcal{H}_1}} ( E\otimes \mathbb{I}_{\mathcal{H}_2}) \rho ( E^\dagger\otimes \mathbb{I}_{\mathcal{H}_2}) = \frac{\mathbb{I}_{1}}{\text{dim}(\mathcal{H}_1)} \otimes \rho_{2}
    \end{equation}
\end{fact}

Where $\mathbb{I}_{1}$ denotes the identity operator on $\mathcal{H}_1$, and $\rho_{2} = \Tr_1 \rho$ is the partial trace.

\begin{fact}
    [Pauli Twirl \cite{Aharonov2008InteractivePF}]\label{fact:twirl} Let $\rho\in \mathcal{L}(\mathcal{H}_1\otimes \mathcal{H}_2)$ be a quantum state, and $F\neq  G\in \mathbb{P}_{\mathcal{H}_1}$. Then, 
    \begin{equation}
       |\mathbb{P}_{\mathcal{H}_1}|^{-1} \sum_{E\in \mathbb{P}_{\mathcal{H}_1}} ( E^\dagger F E\otimes \mathbb{I}_{\mathcal{H}_2}) \rho ( E^\dagger G^\dagger E\otimes \mathbb{I}_{\mathcal{H}_2}) = 0
    \end{equation}
\end{fact}

\subsection{Exact and Approximate Quantum Error Correction}

\subsubsection{Quantum codes, stabilizer codes, and CSS codes}

\begin{definition}\label{definition:quantumcode}
    We refer to a $(n, k)_q$ quantum code as a pair of quantum channels $(\Enc, \Dec)$, where $\Enc:\mathcal{L}(\mathbb{C}^{q^k})\rightarrow \mathcal{L}(\mathbb{C}^{q^n})$, and $\Dec:\mathcal{L}(\mathbb{C}^{q^n})\rightarrow \mathcal{L}(\mathbb{C}^{q^k})$. When $\Enc$ is a unitary, the code-space corresponds to a $q^k$ dimensional subspace $\Pi\subset \mathbb{C}^{q^n}$ of an $n$ qudit Hilbert space. 
\end{definition}

Of particular interest are the class of stabilizer codes, introduced by \cite{Gottesman1997StabilizerCA}. Stabilizer codes are
unitary quantum codes which correspond to the joint eigenspace of a set of commuting Pauli operators, and their encoding unitaries $\Enc$ are Clifford circuits.

\begin{definition}[\cite{Gottesman1997StabilizerCA}]\label{definition:stabilizercode}
    A $[[n, n-r]]_q$ stabilizer code $Q$ is a quantum code corresponding to the joint $+1$ eigenspace of a set $\{S_1, \cdots S_r\}\subset \mathbb{P}_q^n$ of independent, commuting, $n$-qudit Pauli operators.
\end{definition}

The stabilizer group $S(Q)$ is the subgroup of Pauli operators generated by the $S_i$, and the normalizer group $N(Q)$ is the set of Pauli operators which commute with all of $S(Q)$. The elements of $N(Q) - S(Q)$ \footnote{By which we mean the set difference, to distinguish from the quotient group.} are the `undetectable errors' on $Q$, as they map code-states to distinct code-states, and the quotient group $N(Q)/S(Q)$ are the logical operators on $Q$. 

\begin{definition}\label{def:distance}
    The distance $d$ of a stabilizer code $Q$ is the minimum weight of any (undetectable) element in $N(Q) - S(Q)$. If $Q$ encodes $k$ qudits into $n$ over an alphabet size $q$, we refer to $Q$ as $[[n, k, d]]_q$. In addition, the ``pure" distance is said to be the minimum weight of any element in $N(Q)$.
\end{definition}

The class of `detectable' errors of $Q$ are the operators outside $N(Q)$, which we detect by measuring the syndrome vector:

\begin{definition} Let $q=p^m$ be a prime power. For any operator $E\in \mathbb{P}_q^n$, we refer to the syndrome $s_E = (s_1, s_2\cdots, s_r)\in \mathbb{F}_p^{r}$ of the operator $E$ on the stabilizer code $Q$ as the phases $s_i$ defined by $S_i E = \omega^{s_i} ES_i$ for $i\in [r]$. 
\end{definition}

We instantiate most of our constructions in this work using the CSS codes introduced by \cite{Calderbank1996GoodQE, Steane1996SimpleQE}, and their extensions to non-binary fields \cite{Ketkar2006NonbinarySC, Rtteler2004OnQM, Kim2008NonbinaryQE}. 

\begin{definition} [Galois-qudit CSS codes] \label{theorem:galoiscss} 
    Fix two linear classical codes $C_1, C_2 \subset \mathbb{F}_q^n$ of dimension $k$, where $C_2^\perp \subset C_1$ and $C_1, C_2$ both have distance at least $d$. Let CSS$(C_1, C_2)\subset \mathbb{C}^{q^n}$ be the stabilizer code defined by the stabilizers $E_{\textbf{a}, \textbf{b}} \in \mathbb{P}_q^n$, $a\in (C_2)^\perp, b\in (C_1)^\perp$. Then CSS$(C_1, C_2)$ is a $[[n, 2k-n, d]]_q$ stabilizer code. 
\end{definition}

Moreover, the normalizer group of CSS$(C_1, C_2)$ is (up to a global phase) the set of operators $E_{\textbf{a}, \textbf{b}} \in \mathbb{P}_q^n$, for  $a\in C_1, b\in C_2$.

To instantiate our constructions, we require stabilizer codes over binary alphabets of large (near-linear) distance. 

\begin{theorem}
    [Binary Quantum Reed Solomon \cite{Grassl_1999}] \label{theorem:bqrs} For every $0<k<n$, there exists an $[[n, k]]_2$ stabilizer code of distance and pure distance $d = \Omega(\frac{n-k}{\log n})$.
\end{theorem}

\subsubsection{The Erasure Channel and Approximate Quantum Erasure Codes}
\label{section:adaptive}

We formulate the erasure channel following \cite{Grassl2020EntropicPO,Mamindlapally2022SingletonBF}. 

\begin{definition}\label{definition:single-qubit-erasure}
    We define the single-qudit erasure channel $\Res: \mathcal{H}\rightarrow \mathcal{H}\oplus \mathbb{C}\ket{\bot}$ as
    \begin{equation}
        \Res(\rho) = \Tr[\rho]\cdot \ketbra{\bot}
    \end{equation}
\end{definition}

The use of $\ket{\bot}$ indicates the qudit is lost to the decoder, but they can identify its absence.  A quantum code corrects from non-adaptive erasures if it can (approximately) recover every code-state $\rho$ from every one of its (sufficiently large) reduced density matrices $(\mathbb{I}_{[n]\setminus S}\otimes \Res_S)(\rho) = \rho_{[n]\setminus S}\otimes \ketbra{\bot}_S$. \Cref{definition:non-adaptive-aqec} formalizes this notion by measuring the recovery error of the decoding channel in terms of the diamond norm distance to the identity channel:

\begin{definition}\label{definition:non-adaptive-aqec}
    An quantum code $(\Enc, \Dec)$ is a non-adaptive $(\delta, \epsilon)$ approximate quantum erasure code if for all erased subsets $S\subset [n]$ of size $|S|\leq \delta\cdot n$,
    \begin{equation}
        \|\Dec\circ (\mathbb{I}_{[n]\setminus S}\otimes \Res_S) \circ \Enc - \mathbb{I}\|_\diamond \leq \epsilon
    \end{equation}
\end{definition}

We emphasize: in the \textit{non-adaptive} setting, an adversary picks a subset of qubits to erase before receiving the code-state, and the erased locations are known to the decoder. Conversely, in the \textit{adaptive} setting we consider, we allow adversaries to adaptively pick qubits to erase based on the prior qubits they have observed. They do so by applying some adversarial quantum channel $\mathcal{A}$, proceeded by $\Res$ on all the qubits they acted on. Again, this implies the decoder knows which code qubits were observed. 

\begin{definition}\label{definition:aechannel}
    A quantum channel $\mathcal{A}$ is an $(n, \delta)$ adversarial erasure channel if there exists a set $\{A_a\}$ of Krauss operators with supports $S_a = \text{supp}(A_a)\subset [n]$, such that $|S_a|\leq \delta\cdot n:\forall a$ and
\begin{equation}
    \mathcal{A}(\rho) = \sum_a (\mathbb{I}_{[n]\setminus {S_a}}\otimes \Res_{S_a})(A_a\rho A_a^\dagger)
\end{equation}
\end{definition}

Recall that we define the support $S$ of a linear operator $O$ on $n$ qubits to be the subset $S\subset [n]$ s.t. the operator factorizes into a tensor product $O = O_S\otimes \mathbb{I}_{[n]\setminus S}$. \Cref{definition:aechannel} now lets us formalize approximate quantum erasure correction against adaptive quantum adversaries.

\begin{definition}\label{definition:aqec}
    An quantum code $(\Enc, \Dec)$ is an adaptive $(\delta, \epsilon)$ approximate quantum erasure code if for every adversarial $(n,\delta)$ erasure channel $\mathcal{A}$,
    \begin{equation}
        \|\Dec\circ \mathcal{A}\circ \Enc - \mathbb{I}\|_\diamond \leq \epsilon
    \end{equation}
\end{definition}

We remark that this definition implies that after the erasure process and the decoding step, the joint state defined by the decoded message register and the adversary is close to a product state. Moreover, the information-disturbance tradeoff \cite{Kretschmann2006TheIT} implies the state held by the adversary \emph{and} the ancilla qubits during decoding, is approximately independent of the encoded message state. Thereby, even adaptive adversaries learn almost nothing about the message, satisfying a notion of `privacy' akin to secret sharing. 

\subsubsection{Circuit Lower Bounds for Stabilizer codes}

The code-states of quantum error correcting codes, are, generically, highly entangled states. In particular, this suggests that they should be quite ``far" from product states or other states of limited entanglement. A rich line of work in the quantum information community was concerned in proving quantum circuit lower bounds for quantum systems, in particular those based on quantum error correcting codes, which culminated in the recent proof of the NTLS theorem \cite{Eldar2015LocalHW, Nirkhe2018ApproximateLC, Anshu2020CircuitLB, Anshu2022ACO, Anshu2022NLTSHF}.

Here, we leverage a special case of a lemma by \cite{Anshu2020CircuitLB}, which in particular shows that the code-space of stabilizer codes of large distance are exponentially-far from product states in fidelity. 

\begin{lemma}[\cite{Anshu2020CircuitLB}, Lemma 14]\label{lemma:clb}
    If $Q$ is a $[[b, k, d]]_2$ stabilizer code, then no product state $\ket{\phi}=\otimes_{i\in [b]}\ket{\phi_i}$ can have fidelity $f$ with the code-space $\Pi_{Q}$ larger than
    \begin{equation}
       f^2 =  \bra{\phi}\Pi_Q\ket{\phi} \leq 2\cdot 2^{-d^2/(2^{12}\cdot b)}
    \end{equation}
\end{lemma}

\subsection{Non-Malleable Codes}

A key ingredient in our constructions are classical non-malleable codes, and in particular, those secure against bit-wise tampering functions. In this subsection we formally describe these codes and state the known results we require to instantiate our constructions.

\begin{definition}[Non-Malleable Codes \cite{Dziembowski2010}]\label{definition:nm-codes}
A pair of randomized functions $\Enc:\{0, 1\}^{k}\rightarrow \{0, 1\}^{N}$, $\Dec:\{0, 1\}^{N}\rightarrow \{0, 1\}^{k}\cup \{\bot\}$ is said to be a non-malleable code against a class of tampering functions $\mathcal{F}$ with error $\epsilon$ if they satisfy
\begin{enumerate}
    \item (Correctness) $\forall m\in \{0, 1\}^{}, \Dec\circ \Enc(m)=m$.
    \item (Non-Malleability) For every $f\in \mathcal{F}$, there exists a distribution $q_f$ over $\{0, 1\}^{k}\cup \{\bot, \text{same}\}$ satisfying
    \begin{equation}
      \forall m\in \{0, 1\}^{k}: \Dec \circ f \circ \Enc(m) \approx_{\epsilon} p_{f, m}\equiv \begin{cases}
            \text{Sample }\tilde{m}\leftarrow q_f \\
            \text{Output }m \text{ if }\tilde{m}=\text{same, }\tilde{m}\text{ otherwise} 
        \end{cases} 
    \end{equation}

    Where the error is measured in the statistical distance between the distributions.
\end{enumerate} 
\end{definition}

Item (1) is the standard perfect code correctness guarantee. Item (2) stipulates that the outcome of the ``tampering experiment", defined by encoding a message, tampering with it, and decoding it, can be simulated by a distribution which doesn't know the message itself. In other words, the distribution over decoding outputs is a convex combination over the original message, rejection, or an uncorrelated message $\tilde{m}$.\footnote{We spell out this (well-known) interpretation in more detail in \cref{section:auth}. Observe that it enables the adversary to also pick distributions over tampering functions in $\mathcal{F}$, so long as the distribution is independent of $m$.}

As mentioned, we pay particular attention to the well-studied special case where the family of tampering functions $\mathcal{F}\subset \{f: \{0, 1\}^N\rightarrow \{0, 1\}^N\}$ is applied ``bit-wise":

\begin{definition}
    The class of bit-wise tampering functions $\mathcal{F}^{bit}$ on $N$ bits is the set of functions 
    \begin{equation}
        \mathcal{F}^{bit} = \{ (f_1, \cdots, f_N): f_i:\{0, 1\}\rightarrow \{0, 1\}\}
    \end{equation}
\end{definition}

Which, in particular, transforms a codeword $c=(c_1, \cdots, c_N)$ into $(f_1(c_1), \cdots, f_N(c_N))$. To instantiate our constructions, we rely on two constructions of non-malleable codes against bit-wise tampering functions. The first, by \cite{Dziembowski2010}, has constant rate and inverse-exponential error, albeit is randomized. 

\begin{theorem}
    [\cite{Dziembowski2010}] \label{theorem:dpw10} There exists an efficient Monte Carlo construction of non-malleable codes against bit-wise tampering functions for every sufficiently large block-length $N$, of rate $\Theta(1)$, and error $\epsilon\leq 2^{-\Omega(N)}$. The construction succeeds with probability $\geq 1- 2^{-\Omega(N)}$.
\end{theorem}

The second, by \cite{Cheraghchi2013NonmalleableCA}, is explicit and high-rate but suffers from a smaller error. 

\begin{theorem}
    [\cite{Cheraghchi2013NonmalleableCA}]\label{theorem:cg13}  There exists an explicit and efficient family of non-malleable codes against bit-wise tampering functions for every sufficiently large block-length $N$ of rate $\geq 1 - 1/\log N$ and error $\epsilon = 2^{-\Tilde{\Omega}(N^{1/7})}$.
\end{theorem}

\newpage
\section{Pauli Manipulation Detection}
\label{section:PMDs}

In this section, we discuss properties and our explicit construction of Pauli Manipulation Detection (PMD) codes (\Cref{theorem:results-pmdexplicit}). Here, we present a slightly more formal definition:

\begin{definition}\label{def:PMD2}
    A $q^k$-dimensional subspace of $ \mathbb{C}^{q^n}$ with code-space projector $\Pi$ is said to be an $(n, k, \epsilon)_q$-$\PMD$ if, for all non-trivial Pauli operators $E\in \mathcal{P}_q^n\setminus \{\mathbb{I}\}$,
    \begin{equation}
        \|\Pi E\Pi\|_\infty \leq \epsilon
    \end{equation}
\end{definition}

\Cref{def:PMD2} implies that corrupting any code-state $\ket{\psi_1}$ with any non-identity Pauli $E$, leads to a state which is near-orthogonal to the code-space:
\begin{equation}
   |\bra{\psi_1} E^\dagger \ket{\psi_2}| = |\bra{\psi_1} E^\dagger \Pi\ketbra{\psi_2} \Pi E\ket{\psi_1}|^{1/2} \leq \|\Pi E^\dagger\Pi E\Pi\|_\infty^{1/2} = \|\Pi E\Pi\|_\infty \leq \epsilon
\end{equation}
For all code-states $\ket{\psi_2}\in \Pi$. In a sense, this suggests that one can check if a given code-state was perturbed by a Pauli, without corrupting it. We make this connection precise in \Cref{claim:pmdapproximation}. Subsequently, we prove our main result on the construction of PMDs from ``pairwise-detectable" sets of stabilizer codes.

\begin{theorem} (Restatement of \Cref{theorem:results-pmdexplicit})\label{thm:pmds}
    For every prime power $q$, sufficiently large $n\in \mathbb{N}$, and $\lambda | n$, there exists an explicit and efficient $(n, n-2\lambda, \epsilon)_q$-$\PMD$ with error $\epsilon \leq 2\cdot n^{1/2}\cdot q^{-\lambda/2}$. 
\end{theorem}

\subsection{Robust Manipulation Detection}

To begin, in \Cref{claim:pmdapproximation} we formalize the property that $\PMD$s detect errors without disturbing the corrupted code-state. Let $\Enc$ be the unitary which encodes into the $\PMD$ code, $\Enc \ket{\psi}\ket{0^{\Anc}} = \Bar{\ket{\psi}}$, using $|\Anc|=n-k$ qudit ancillas. We show that there exists a $n$ qudit $+1$ qubit unitary $\Auth$ which either exactly accepts, or approximately rejects code-states corrupted by Paulis. The extra qubit in register $F$ acts as a ``flag", indicating acceptance/rejection:

\begin{claim} \label{claim:pmdapproximation}
    If $\Pi$ is an $(n, k, \epsilon)$-$\PMD$, then there exists a unitary $\Auth:\mathbb{C}^{q^n}\otimes \mathbb{C}^2\rightarrow \mathbb{C}^{q^ n}\otimes \mathbb{C}^2$, satisfying
    \begin{enumerate}
        \item (Correctness) $\Auth$ exactly recovers code-states,  $\forall\ket{\psi}\in \mathbb{C}^{q^k}$: 
        \begin{equation}
            \Auth \big(\Enc \ket{\psi}\ket{0^{\Anc}}\big)\ket{0}_F = \ket{\psi}\ket{0^{\Anc}}\ket{1}_F.
        \end{equation}
        \item (Robust Error-Detection) For every Pauli $E\in \mathbb{P}_q^n
\setminus \{\mathbb{I}\}$ and $k$ qudit states $\psi$, let $\ket{\phi_E} = E\cdot \Enc\ket{\psi}\otimes \ket{0^{\Anc}}\otimes \ket{0}_F$ correspond to a corrupted code-state. Then,
        \begin{equation}
    \|\Auth\ket{\phi_E}\ket{0}_F- \ket{\phi_E} \ket{0}_F  \|_2\leq \sqrt{2}\cdot \epsilon
\end{equation}
    \end{enumerate}
\end{claim}

At an intuitive level, item (1) above simply stipulates that if there is no corruption, then we recover the original message state. Item 2 indicates that $\Auth$ detects the Pauli error $E\neq \mathbb{I}$ while approximately preserving the corrupted code-state $\ket{\phi_E}$. 

\begin{proof}

To define the authentication unitary $\Auth$, we first measure $\Pi$ coherently, and write out the outcome on the binary side-register (the flag $F$). Conditioned on the measurement outcome, we revert the encoding $\Enc^\dagger$ on the code-register. Formally, $\Auth = C(\Enc^\dagger) U$, where $U$ performs the measurement, and $C(\Enc^\dagger)$ performs the controlled gate
\begin{equation}
    U = \Pi\otimes X_F + (1-\Pi)\otimes Z_F , \text{ and } C(\Enc^\dagger) = \Enc^\dagger_Q\otimes \ket{1}\bra{1}_F + \mathbb{I}_Q\otimes  \ket{0}\bra{0}_F
\end{equation}

Correctness is immediate. To prove the error-detection property, note that $\ket{\phi_E}\ket{0}_F = E\cdot \Enc(\ket{\psi})\ket{0}_F$ and $\Auth\ket{\phi_E}\ket{0}_F$ are pure states, where the real part of their inner product satisfies
\begin{equation}
\mathcal{R}e\bigg[\bra{\phi_E}\bra{0}_F\Auth\ket{\phi_E}\ket{0}_F\bigg] = \bra{\psi} \Enc^\dagger E^\dagger (\mathbb{I}-\Pi)E \Enc \ket{\psi}  \geq 1 - \|\Pi E \Pi E\Pi \|_\infty = 1- \|\Pi E\Pi \|_\infty^2\geq 1- \epsilon^2
\end{equation}

The relation $\|x-y\|_2 = \sqrt{2(1-\mathcal{R}e[x\cdot y])}$ gives us the desired bound. 

\end{proof}

\subsection{PMDs from Pairwise-Detectable Purity Testing}

In this subsection we describe a construction of PMDs based on the PTC's of \cite{barnum2002authentication}, and prove \Cref{thm:pmds}. Recall the definition of PTCs and our notion of pairwise detectability:

\begin{definition} [Purity Testing]
    A set of stabilizer codes $\{Q_k\}_{k\in K}$ is called an $\epsilon$-strong PTC if for all non-trivial Pauli operators $E\in \mathbb{P}_q^n\setminus \{\mathbb{I}\}$, 
    \begin{equation}
        \mathbb{P}_{k\leftarrow K}\big[ E\in N(Q_k)\big] \leq \epsilon
    \end{equation}
\end{definition}

In other words, the probability that $E$ is detectable is $\geq 1-\epsilon$, for a random choice of code $k\leftarrow K$. \cite{barnum2002authentication} presented a particularly key-efficient family of PTC's, and used them to develop quantum message authentication codes:

\begin{theorem}[\cite{barnum2002authentication}]
For every prime power $q$ and sufficiently large $n, \lambda\in \mathbb{N}, \lambda|n$, there exists a set of  $[[n, n-\lambda]]_q$ stabilizer codes $\{Q_\alpha\}_{\alpha\in \mathbb{F}_{q^\lambda}}$ which form an $\epsilon$-strong PTC with error $\epsilon \leq n\cdot q^{-\lambda}$.
\end{theorem}

\begin{definition}[Pairwise Detectability]
    A set of stabilizer codes $\{Q_k\}_{k\in K}$ is called $\delta$ pairwise-detectable if, for all $s\neq 0$, 
    \begin{equation}
        \mathbb{P}_{k\leftarrow K}\big[S(Q_k) \cap N(Q_{k+s}) \neq \{\mathbb{I}\}\big] \leq \delta
    \end{equation}
\end{definition}

In the next subsection, we show the family of PTC's in \cite{barnum2002authentication} is pairwise-detectable. For now, we simply import the result and use it to construct PMDs. 

\begin{lemma} [\Cref{claim:detectibleproof}, restatement]\label{lemma:pairwisebarnum}
    The family of $[[n, n-\lambda]]_q$ PTC's considered by \cite{barnum2002authentication} are $\delta$ pairwise-detectable with $\delta \leq 2\cdot n \cdot q^{-\lambda}$. 
\end{lemma}

We define our PMDs as follows. If each $Q_k$ has a unitary encoding map $\Enc_k$ acting on some number $a$ of ancillas and the message state $\ket{\phi}$, then our PMDs encode $\ket{\phi}$ as superpositions of PTC encodings, where a ``key register'' $S$ holds $k$, and a ``code register'' $C$ holds $\Enc_k \ket{\phi}\ket{0^{\Anc}}$. 

\begin{equation}
    \Enc  = \sum_k \ket{k} \bra{k}_S H_S\otimes \Enc_k, \Enc(\ket{\psi}) = |K|^{-1/2}\sum_k \ket{k}\otimes \Enc_k (\ket{\phi})
\end{equation}

where $H_S$ prepares the uniform superposition over keys, $H_S\ket{0^{key}} = |K|^{-1/2}\sum_k \ket{k}$. There is a natural projection $\Pi$ onto the codespace,

    \begin{equation}
        \Pi = \Enc (\mathbb{I}\otimes 0^{\Anc+key})\Enc^{\dagger} = |K|^{-1}\sum_{k, k'}\ket{k}\bra{k'}\otimes \Enc_k (\mathbb{I}\otimes 0^{\Anc})\Enc_{k'}^\dagger
    \end{equation}

    Recall that $\Pi_k= \Enc_k (\mathbb{I}\otimes 0^{\Anc})\Enc_{k}^\dagger$ is simply the projector onto $Q_k$. We readily inspect $\Pi^2=\Pi$, and $\Pi\Enc\ket{\phi} = \Enc\ket{\phi}$ for all $\phi$.

\begin{lemma}\label{lemma:pmds-from-ptcs}
    If a set of $[[n, n-\lambda]]_q$ stabilizer codes $\{Q_k\}_{k\in K}$ is a $\epsilon$-strong PTC and is $\delta$ pairwise-detectable, then the code defined by $\Enc, \Pi$ is a $(n+\log_q K, n-\lambda, \Delta)_q$-$\PMD$ with $\Delta\leq  \max(\epsilon, \sqrt{q^{-\lambda} + \delta})$
\end{lemma}

We obtain \Cref{thm:pmds} by instantiating the lemma above with the PTCs in \cite{barnum2002authentication} and \Cref{lemma:pairwisebarnum}. We divide the proof of this lemma in two steps. First, in \Cref{claim:phase}, we show that the probability any Pauli error $E$ is undetected is roughly related to the key manipulation security experiment we devised in \Cref{section:overview-explicit}. Then, in \Cref{claim:ptc-a's}, we show that pairwise-detectability implies that our candidate PMD doesn't fail the experiment except for assymptotically small probability.

\begin{claim} \label{claim:phase}
    Let $E = E_{K}\otimes E_{C}\neq \mathbb{I}_K\otimes \mathbb{I}_C$ be a Pauli error, where $E_{K} = X^aZ^b$ is supported on the key register $K$ and $E_{C}$ on the code $C$. 
    \begin{equation}
            \|\Pi E\Pi\|_\infty \leq |K|^{-1}\sum_{k} \|\Pi_{k}E_C\Pi_{k+a}\|_\infty
    \end{equation}
    Moreover, if $E_C = \mathbb{I}$ and $E_K = Z^b$ is a phase, then $\|\Pi E\Pi\|_\infty = 0$.
\end{claim}

\begin{claim}\label{claim:ptc-a's}
    If $a=0$ and $E_C\neq \mathbb{I}$, $\|\Pi E\Pi\|_\infty\leq \epsilon$. Moreover, if $a\neq 0$, then $ \|\Pi E\Pi\|_\infty\leq \sqrt{q^{-\lambda}+\delta}$
\end{claim}

Combined, \Cref{claim:phase} and \Cref{claim:ptc-a's} show that $ \|\Pi E\Pi\|_\infty\leq \epsilon$ for all $E\neq \mathbb{I}$, and thus we have proved \Cref{lemma:pmds-from-ptcs}. It remains now to prove the claims:

\begin{proof}

[of \Cref{claim:phase}]
    Let $E = E_S\otimes E_C\neq \mathbb{I}$ be any Pauli operator acting on the key $S$ and code register $C$. We begin by writing out $\Pi E\Pi$:
    \begin{equation}
       \Pi E\Pi = |K|^{-2} \sum_{k_1, k_2, k_3, k_4\in K}\bra{k_2} E_S\ket{k_3} \cdot \ket{k_1}\bra{k_4} \otimes \Enc_{k_1} \Enc_{k_2}^\dagger \Pi_{k_2}E_C\Pi_{k_3}\Enc_{k_3} \Enc_{k_4}^\dagger
    \end{equation}

We make the explicit assumption that the Pauli operators acting on the key register $S$ obey the additive group structure of the set of PTC keys $K$. That is $X^aZ^b\ket{k}_S = \omega^{i b\cdot k}\ket{k+a}_S$, and $\sum_{z\in K} \omega^{i b\cdot z} = |K|\cdot \delta_{b, 0}$. 
    
    First, if $E_C=\mathbb{I}$ and $E_K$ is a phase error, then $\bra{k_2} Z^{e}\ket{k_3} = \delta_{k_2, k_3}\cdot w^{i b\cdot k_3}$, and $b\neq 0$. Note that $E_C=\mathbb{I}$ implies $\Enc_{k_2}^\dagger \Pi_{k_2}E_C\Pi_{k_2}\Enc_{k_2} = (\mathbb{I}\otimes 0^\Anc)$. Thus, summing over $k_2\in K$,
    \begin{gather}
        \Pi E\Pi = |K|^{-2} \bigg(\sum_{k_2} w^{i b\cdot k_2}\bigg) \cdot \sum_{k_1, k_4\in K} \cdot \ket{k_1}\bra{k_4} \otimes 
\Enc_{k_1}(\mathbb{I}\otimes 0^\Anc) \Enc_{k_4}^\dagger  = 0. 
    \end{gather}

    To conclude this claim, if $E_K = X^aZ^b$, we use the unitary invariance of $\|\cdot\|_\infty$, the factorization of the norm of tensor products, and the triangle inequality, to write
    \begin{gather}
        \|\Pi E\Pi\|_\infty \leq |K|^{-2} \bigg\| \sum_{k_1, k_4\in K}  \ket{k_1}\bra{k_4} \bigg\|_\infty \bigg\|\sum_{k_2\in K} \omega^{i(k_2-a)\cdot b}\Enc_{k_2}^\dagger \Pi_{k_2}E_C\Pi_{k_2-a}\Enc_{k_2-a}\bigg\|_\infty \\ 
        \leq |K|^{-1} \sum_{k_2} \bigg\|\Pi_{k_2}E_C\Pi_{k_2-a}\bigg\|_\infty
    \end{gather}

    Where we used $\| \sum_{x, y\in K}  \ket{x}\bra{y}\|_\infty = |K|$. 
\end{proof}    

\begin{proof} 

[of \Cref{claim:ptc-a's}]
    If $a=0$ but $E_C\neq \mathbb{I}$, then we have 
    \begin{equation}
            \|\Pi E\Pi\|_\infty \leq |K|^{-1}\sum_{k} \|\Pi_{k}E_C\Pi_{k}\|_\infty = \mathbb{P}_{k\leftarrow K}[E\in Q_k^\dagger] \leq \epsilon
    \end{equation}

    If $a\neq 0$, 
    \begin{gather}
    \|\Pi_{k}E_C\Pi_{k+a}\|_\infty^2 = \|\Pi_{k}E_C\Pi_{k+a}E_C^\dagger \Pi_{k}\big\|_\infty  \leq |S_{k+a}|^{-1} \sum_{\sigma\in S_{k+a}\cap Q_{k}^\perp} \|\sigma\|_\infty \\ \leq |S_{k+a}|^{-1}\cdot |S_{k+a}\cap Q_{k}^\perp| = q^{-\lambda}+\mathbb{I}[S_{k+a} \cap Q_{k}^\perp \neq \{\mathbb{I}\}]
\end{gather}

Since $\{Q_k\}$ is $\delta$ pairwise-detectable, in expectation over $k$ this is simply $1/q^\lambda + \delta$. By Jensen's inequality, 
\begin{equation}
    \|\Pi E\Pi\|_\infty \leq \sqrt{q^{-\lambda} + \delta}
\end{equation}
   
\end{proof}

\subsection{The \texorpdfstring{\cite{barnum2002authentication}}{BCG+02} Purity Testing Codes are Pairwise-Detectable}
\label{subsection:PTCs}

To begin, let us review the construction of PTCs by \cite{barnum2002authentication}:

\cite{barnum2002authentication} define a set of $[[n, n-\lambda]]_{p^m}$ stabilizer codes as follows. Assume $\lambda|n$, and let $r = n/\lambda$. Consider the following set of $q^\lambda$ vectors over $\mathbb{F}_{q^\lambda}^{2\cdot r}$:

\begin{equation}
    V = \{v_\alpha = (1, \alpha, \alpha^2\cdots \alpha^{2r-1}): \alpha\in \mathbb{F}_{q^\lambda}\} 
\end{equation}

Each vector $v \in V$ defines a one-dimensional subspace of the $2r$-dimensional vector space over $\mathbb{F}_{q^\lambda}$, defined by the scalar multiplication $\{\gamma\cdot  v\}_{\gamma\in \mathbb{F}_{q^\lambda}}$. These 1D subspaces over $\mathbb{F}_{q^\lambda}$ can be considered as $\lambda$-dimensional subspaces over $\mathbb{F}_{q}$ (and $\lambda\cdot m$-dimensional over $\mathbb{F}_p$). Concretely, let us partition $v\in \mathbb{F}_{q^\lambda}^{2\cdot r}$ as $v = (a, b)$, $a, b\in \mathbb{F}_{q^\lambda}^{r}$ into ``X" and ``Z" halves, and use a self-dual pair of basis $\phi_1, \phi_2:\mathbb{F}_{p^{m\lambda}}\rightarrow \mathbb{F}_{p}^{\lambda\cdot m}$ defined in Preliminaries to explicitly define the subspace $V_v\subset \mathbb{F}_p^{2n\cdot m}$:
\begin{equation}
    V_v = \big\{ \phi(\gamma v) =(\phi_1(\gamma a_1),  \cdots \phi_1(\gamma a_r), \phi_2(\gamma b_1), \cdots, \phi_2(\gamma b_r) )\in \mathbb{F}_q^{2n}: \gamma \in \mathbb{F}_{q^\lambda}\big\}
\end{equation}

\cite{barnum2002authentication} point out that for each $v\in V$, the subspace $V_v$ can be used to define a set of commuting operators, which in turn define a stabilizer code $Q_v$. For completeness, we spell out the operator definitions here. Consider the set $S_v$ of $n\cdot m$ qudit Pauli operators $E_{x, y} = X^xZ^y\in \mathbb{P}_p^{nm}$, for each $(x, y)\in V_v$. If $\phi(\gamma v)=(x, y), \phi(\gamma' v)=(x', y')$, then the symplectic form 

\begin{equation}
   B(\gamma v, \gamma' v) =  \langle x, y'\rangle-\langle x', y\rangle = \text{Tr}_{\mathbb{F}_{p^{m\lambda}}\rightarrow \mathbb{F}_{p}}\big[\gamma \gamma' (a\cdot b -a\cdot b) ] = 0,
\end{equation}

and thus $ E_{x, y} ,E_{x', y'}$ commute. Thereby, $S_v$ form a set of stabilizers for a $[[n\cdot m, n\cdot m-\lambda\cdot m]]_p$ code $Q_v$, which can be considered as a $[[n, n-\lambda]]_{p^m}$ code over $q=p^m$. For notational convenience, from here onwards we will index this set of codes $Q_{v_\alpha}, v_\alpha\in V$, as $Q_\alpha$ for $\alpha\in \mathbb{F}_{q^\lambda}$.

\begin{theorem}[\cite{barnum2002authentication}]
    The set of stabilizer codes $\{Q_\alpha\}_{\alpha\in \mathbb{F}_{q^\lambda}}$ is an $\epsilon$-strong PTC with error $\epsilon \leq n\cdot q^{-\lambda}$
\end{theorem}

\subsubsection{Proof of \Cref{lemma:pairwisebarnum}}

To devise our PMD codes, we make the observation that the subspaces $Q_\alpha$ satisfy a restricted form of orthogonality. Let us consider two such codes $Q_\alpha, Q_{\alpha+s}$, and let $S_\alpha, S_{\alpha+\beta}$ be their corresponding stabilizer groups, for $\alpha, \beta\in \mathbb{F}_{q^{\lambda}}$.

\begin{claim} \label{claim:detectibleproof}
Fix $\beta\in \mathbb{F}_{q^{\lambda}}$, and sample $\alpha\in \mathbb{F}_{q^{\lambda}}$ uniformly at random. The probability a stabilizer of $Q_\alpha$ is a logical operator of $Q_{\alpha+\beta}$, or equivalently,

    \begin{equation}
        \mathbb{P}_{\alpha\in \mathbb{F}_{q^\lambda}} \big [ S_\alpha\cap Q_{\alpha+\beta}^\perp \neq \mathbb{I} \big] \leq 2\cdot n\cdot q^{-\lambda}.
    \end{equation}

    Thereby $\{Q_k\}$ is $2\cdot n\cdot q^{-\lambda}$ pairwise-detectable. 
\end{claim}

\begin{proof}
    First, we prove that for the family of stabilizer codes defined above, $S_a\cap Q_{b}^\perp \neq \mathbb{I}$ if all of $S_a$ commutes with all of $S_{b}$. Then, we show that $S_\alpha$ commutes with $S_{\alpha+\beta}$ if $\alpha$ is the solution to a certain low-degree polynomial, and an application of Schwartz-Zippel gives us the desired bound.  

    If $S_a\cap Q_{b}^\perp \neq \mathbb{I}$, then, by writing the commutation relations of $S_a, S_b$ in terms of the symplectic form $B$, there exists a $\gamma\in \mathbb{F}_{q^\lambda}$ such that $B(\gamma v_a, \gamma' v_b)=0$ for all $\gamma'\in \mathbb{F}_{q^\lambda}$. However, bilinearity implies $B(\gamma v_a, \gamma' v_b) = B( \gamma' \gamma v_a, v_b)$. Thus $B( \gamma v_a,  \gamma' v_b) = 0$ for all $\gamma, \gamma'\in \mathbb{F}_{q^\lambda}$, which implies that $S_a$ commutes with $S_b$, and moreover the symplectic inner product over $\mathbb{F}_{q^\lambda}$

    \begin{equation}
        \langle v_a, v_b\rangle_s = (b^r-a^r)\sum_{i=0}^r (ab)^i = (b^r-a^r) ((ab)^{r+1}-1)/(ab-1) = 0
    \end{equation}

    Let us now fix $\beta\neq 0$, such that $S_\alpha, S_{\alpha+\beta}$ commute if $\alpha$ is a root to the univariate polynomial

    \begin{equation}
     p_\beta(\alpha) =   \big( (\alpha+\beta)^r-\alpha^r\big) \cdot \big(\alpha^{r+1}(\alpha+\beta)^{r+1}-1\big) 
    \end{equation}

    $p_\beta(\alpha)$ has degree $\leq 3r+2$, and thereby by Schartz-Zippel the probability over random $\alpha\in \mathbb{F}_{q^\lambda}$ that $S_\alpha, S_{\alpha+\beta}$ commute is $\leq (3r+2)\cdot q^{-\lambda}\leq 4rq^{-\lambda} \leq 2nq^{-\lambda}$

\end{proof}

\newpage

\section{Quantum List Decoding from Erasures}
\label{section:QLDE}

In this section, we discuss a variant of the list-decodable quantum codes studied by \cite{Leung2006CommunicatingOA, Bergamaschi2022ApproachingTQ}. Namely, stabilizer codes which are list-decodable from erasures, following the notion for classical codes introduced by \cite{Guruswami2001ListDF}. 

We begin in \cref{subsection:qlde-css}, where we revisit the classical definition, introduce our new definition of quantum list-decoding from erasures (QLDE), and show how to construct such codes from CSS codes. Then, in \cref{section:qlde-randomcss}, we show that random CSS codes are QLDE with high probability. Finally, in \cref{subsection:qlde-operational}, we describe an operational interpretation to QLDE as collapsing the quantum erasure channel into a sparse Pauli channel (on a known location of the code). We refer the reader to \cref{section:prelim} for the relevant background on stabilizer codes and the erasure channel.

\subsection{Definitions, and constructions from CSS codes}
\label{subsection:qlde-css}

Let us review the classical notion of list decoding from erasures:

\begin{definition}\label{definition:LDE}
    A code $C\subset \Sigma^n$ is $(\delta, L)$-LDE (list-decodable from erasures) if $\forall y\in \Sigma^{ n - t}$ and $T\subset [n], |T|=n-t\geq (1-\delta)\cdot n$, 
    \begin{equation}
        \big|\{c\in C: c_T = y\}\big|\leq L
    \end{equation}
\end{definition}

That is, at most $L$ codewords agree with each sub-string $y$. To understand our quantum analog, let us consider linear codes. Linear $(\delta, L)$-LDE have the remarkable property that they are efficiently list decodable (in time polynomial in $n, L, \log q$), since computing the list corresponds to solving the linear system in \Cref{definition:LDE}.

Moreover, when $C$ is linear, one can cast \Cref{definition:LDE} in terms of the syndromes of additive errors supported on the erased subset:

\begin{claim}\label{claim:syndromelinear}
    If a linear code $C\subset \Sigma^n$ of parity check matrix $H$ is $(\delta, L)$-LDE, then for all syndrome vectors $s$ and subsets $T\subset [n], |T|=n-t\geq (1-\delta)\cdot n$, 
    \begin{equation}
        \big|\{e\in \Sigma^n: e_T = 0 \text{ and } s = He\}\big|\leq L
    \end{equation}
\end{claim}

\begin{proof}
    Assume there exists a vector $e\in \Sigma^n$ with $e_T=0$ and syndrome $s = He$. By linearity, every vector $a\in \Sigma^n$ of syndrome $s = Ha$ can be written as $e-c$, for some $c\in C$. However, $0=a_T = e_T-c_T = -c_T$, so
    \begin{equation}
        \big|\{y\in \Sigma^n: y_T = 0 \text{ and } s = Hy\}\big| = \big|\{c\in C: c_T = 0\}\big|\leq L
    \end{equation}
\end{proof}

With this \textit{syndrome-based} definition in mind, we can now formulate a notion of list decoding from erasures in the quantum setting. Recall the notion of independent, or logically-distinct operators in a stabilizer code:

\begin{definition}
    If $Q$ is a stabilizer code with stabilizer group $\mathcal{S}$, then we refer to a pair of Pauli operators $O, O'\in \mathcal{P}_q^n$ as $\mathcal{S}$-logically-equivalent if $O'$ is in the coset $O \mathcal{S}$.
\end{definition}

We say $O, O'$ are $\mathcal{S}$-logically-distinct otherwise. Typically the stabilizer group $\mathcal{S}$ is implicit, and we drop it. We observe that if $O, O'$ are logically-equivalent, then $O\ket{\psi} \propto O'\ket{\psi}$ for any code-state $\ket{\psi}\in Q$. 

Qualitatively, the definition of erasure list decoding for stabilizer codes we study stipulates that at most $L$ logically-different Pauli's can have the same small support and the same syndrome. By design, each of these operators ``corrects'' a corrupted codeword back into the codespace.

\begin{definition}
    \label{def:QLDE}
    An $[[n, k]]_q$ stabilizer code $Q$ is $(\delta, L)$-QLDE (quantum list-decodable from erasures) if for all syndrome vectors $s$ and subsets $T\subset [n], |T| \geq (1-\delta)\cdot n$, there are at most $L$ logically-distinct Pauli's $E\in \mathcal{P}_q^n$ supported on $[n]\setminus T$ and of syndrome $s$.
\end{definition}

We remark this definition is equivalent to a constraint on subgroups of the normalizer group of $Q$. For every subset $W\subset [n]$, let $S_W$ be the subgroup of the stabilizer group $S$ supported only on the qudits in $W$, and let $N_W$ be the subgroup of the normalizer group supported only on $W$. Then, 

\begin{claim}\label{claim:qldenormalizer}
    An $[[n, k]]_q$ stabilizer code $Q$ is $(\delta, L)$-QLDE if for every subset $T\subset [n], |T| \leq \delta\cdot n$, the quotient group $|N_{ T} / S_{ T}|\leq L$.
\end{claim}

\begin{proof}
    If any error $E$ supported on $T$ has syndrome $s$, then the coset of $\leq L$ operators $EA$, $A\in N_{ T} / S_{ T}$, all have syndrome $s$. These operators are all logically-distinct, and moreover, any other error $F$ of syndrome $s$ is logically-equivalent to an operator in that coset.
\end{proof}

We note that list-decoding stabilizer codes from erasures thus amounts to finding a basis of generators for $N_{T}, S_T, N_{ T} / S_{ T}$ and any error $E$ of syndrome $s$. By representing these operators as vectors over $\mathbb{F}_q^{2n}$, one can compute the candidate list of correction operators in time poly$(n, L, \log q)$. In this manner, we essentially inherit the efficiency of linear classical codes. Before discussing the operational applications of \Cref{def:QLDE}, we show how to construct these codes via a connection to CSS codes over finite fields, akin to \cite{Bergamaschi2022ApproachingTQ}:

\begin{claim}
    If $C_1, C_2$ are two $\mathbb{F}_q$-linear, $(\delta, L)$-LDE classical codes where $C_2^\perp\subset C_1$, then CSS$(C_1, C_2)$ is $(\delta, L^2)$-QLDE. 
\end{claim}

\begin{proof}
    Let $C_1, C_2$ and CSS$(C_1, C_2)$ have block-length $n$, and fix any set $T\subset [n]$ of size at least $(1-\delta)\cdot n$. By \Cref{claim:syndromelinear}, since $C_1, C_2$ are $\mathbb{F}_q$-linear and $(\delta, L)$-LDE, then there exist at most $L^2$ pairs of additive errors $a, b$ supported only on $[n]\setminus T$, which match any pair of syndromes $s_x, s_z$. Thus, CSS$(C_1, C_2)$ is $(\delta, L^2)$-QLDE. 
\end{proof}

In \Cref{section:qlde-randomcss}, we use the above to construct CSS codes over qubits which are erasure list-decodable up to the quantum singleton bound, i.e. $[[n, r\cdot n]]_2$ stabilizer codes which are $((1-r-\epsilon)/2, 2^{O(1/\epsilon)})$-QLDE.

\subsection{Properties of random CSS codes}
\label{section:qlde-randomcss}

In this section, we show that random CSS codes define quantum codes which are list-decodable from erasures up to the singleton bound. 

We generate a random $[[n, k]]_q$ CSS code as follows. First, sample  $k_1 = (n+k)/2$ random linearly independent vectors $g_1,\cdots g_{k_1}$ in $\mathbb{F}_q^n$. Let $G_1 = (g_1\cdots g_{k_1}) \in \mathbb{F}_q^{n\times k_1}$ be the matrix defined by the sampled vectors and $H_2^T = (g_1\cdots g_{k_2})$ to be the first $k_2 = n-k_1 $ columns of $G_1$. We define $C_1$ to be the classical linear code defined by using $G_1$ as a generator matrix, and $C_2$ to be the classical linear code defined by using $H_2$ as the parity check matrix, such that $C_2^\perp\subset C_1$ by construction. 

We now consider the random Galois-Qudit CSS code $Q_R = $ CSS$(C_1, C_2)$, where $R=k/n$. We first reason that this code has rate $R$ with high probability, so long as $H_2, G_1$ are full rank. A standard linear algebra computation gives: 

\begin{claim}[See, e.g., \cite{Ding2014ErasureLC}]
    $Q_R$ has rate $R$ with probability $\geq 1-4\cdot n \cdot q^{-n(1-R)/2}$.
\end{claim}

Leveraging the following classical result by \cite{Ding2014ErasureLC, Guruswami2001ListDF}, 

\begin{theorem}[\cite{Ding2014ErasureLC, Guruswami2001ListDF}]
\label{theorem:randomld}
For every $\gamma>0$, $0<R<1$ and sufficiently large $n$, with probability at least $1-q^{-n}$, a random linear code over $\mathbb{F}_q$ of length $n$ and rate $R$ is $(1-R-\gamma, 2^{2/\gamma+1})$-LDE.
\end{theorem}

Thus, by a union bound over $C_1, C_2$,

\begin{corollary}\label{corollary:randomcss}
For every $\gamma>0$, $0<R<1$, and sufficiently large $n$, with probability at least $\geq 1-5\cdot n \cdot q^{-n(1-R)/2}$, the random Galois-Qudit CSS code $Q_R$ has rate $R$ and is $(\frac{1}{2}(1-R-\gamma), 2^{8/\gamma})$-QLDE. 
\end{corollary}

We remark that this statement is vacuous unless $R+\gamma \leq 1$, and thus the failure probability is always $\leq 2^{-\Omega(\gamma\cdot n)}$ if, say $\gamma\geq n^{-0.99}$.

\subsection{Collapsing erasure errors into Sparse Pauli Channels}
\label{subsection:qlde-operational}

Let us now discuss an operational interpretation to \Cref{def:QLDE}. Consider receiving a reduced density matrix $\rho_T = \text{Tr}_{[n]\setminus T}[\psi]$ of a code-state $\psi$ of a $(\delta, L)$-QLDE stabilizer code, and let us measure its syndrome. To do so, we append $n-t$ maximally mixed states, and apply a projective measurement $\{\Pi_s\}$. We remark that receiving the RDM $\rho_T$ is equivalent to applying the completely depolarizing channel to $[n]\setminus T$, 

\begin{equation}
    \rho_T\otimes \mathbb{I}_{[n]\setminus T} = (\mathbb{I}_T\otimes \mathcal{D})(\psi) = q^{-2(n-t)}\sum_{\sigma\in \mathcal{P}_q^{n-t}} (\mathbb{I}_T\otimes \sigma) \psi(\mathbb{I}_T\otimes \sigma^\dagger)
\end{equation}

and thereby, the syndrome measurement collapses the error into a mixture of $\leq L$ logically distinct terms (a sparse Pauli channel):
\begin{gather}
   \Pi_s \big(\rho_T\otimes \mathbb{I}_{[n]\setminus T} \big)\Pi_s = q^{-2(n-t)} \sum_{\substack{\sigma \in \mathcal{P}_q^{n-t} \\ \sigma \text{ of syndrome }s} }(\mathbb{I}_T\otimes \sigma) \psi(\mathbb{I}_T\otimes \sigma^\dagger)
\end{gather}

More generally, one can consider arbitrary channels $\mathcal{A}$ acting on fixed subsets of qudits:

\begin{claim}\label{claim:syndromecollapse}
    Let $\psi$ be a code-state of a $(\delta, L)$-QLDE stabilizer code, and $\mathcal{A} = \{A_\mu\}$ be a CP map supported on a fixed subset of $\leq \delta$ fraction of symbols of $\psi$, with Kraus operators $A_\mu$. Then, measuring the syndrome of $\mathcal{A}(\psi)$ collapses the state into superpositions of $\leq L$ Pauli errors, i.e.

    \begin{equation}
        \Pi_s A_\mu\psi A_\mu^\dagger \Pi_s = \bigg(\sum_{i\in [L]} a_i \sigma_i \ket{\psi}\bigg)  \bigg(\bra{\psi}\sum_{i\in [L]} a_i^* \sigma_i^\dagger\bigg)
    \end{equation}

    Where each Pauli $\sigma_i$ has syndrome $s$, is supported on the same set of $\leq \delta$ fraction of qudits, and are all pairwise logically-distinct. 
\end{claim}

\begin{proof}
    Let us decompose each Kraus operator $A_\mu = \sum a_\sigma \sigma $ into a Pauli basis. We observe
    \begin{equation}
        \Pi_s \sigma \ket{\psi} = \sigma \ket{\psi} \text{ if } \sigma \text{ has syndrome } s, \text{ and otherwise } \Pi_s \sigma \ket{\psi} = 0.
    \end{equation}

    Since $\mathcal{A}$ has relative support size $\leq \delta$ and the stabilizer code is $(\delta, L)$-QLDE, \Cref{def:QLDE} tells us at most $L$ logically-distinct Paulis have $\Pi_s \sigma \ket{\psi} \neq 0$.
\end{proof}

\newpage

\section{Approximate Quantum Erasure Correction}
\label{section:AQEC}

In this section, we prove our main result on the application of PMD codes to approximate quantum erasure correction (\cref{theorem:main}). 

We begin in \Cref{section:aqec-pmd}, where we describe our code construction from PMDs and stabilizer codes which are list-decodable from erasures. In the remaining subsections (\cref{section:aqecc-analysis} and \cref{subsection:aqecc-proofs}), we analyze and prove the theorem. We refer the reader to \Cref{section:adaptive} for formal definitions of the adaptive \& non-adaptive erasures model and approximate quantum error correction, and to \cref{section:QLDE} for formal definitions and constructions of quantum codes which are list-decodable from erasures.

\subsection{The Code Construction}
\label{section:aqec-pmd}

\subsubsection{Encoding}

We combine 

\begin{itemize}
    \item[--] An $(m, k, \epsilon)_q$ $\PMD$ code $\Pi$, with associated encoding and authentication unitaries $(\Enc_{\PMD}, \Auth_{\PMD})$ (See \cref{def:PMD2} and \cref{claim:pmdapproximation}). In particular, the $\PMD$ code of \cref{theorem:results-pmdexplicit}.
    \item[--] An $[[n, m]]_q$ stabilizer code $Q$ with encoding unitary $\Enc_Q$, which is $(\delta, L)$-QLDE. In particular, that of \cref{corollary:randomcss}.
\end{itemize}

We encode any $k$ qudit message state $\ket{\psi}$ first with the encoding unitary $\Enc_{\PMD}$, and then with $\Enc_Q$: 
\begin{equation}
  \Enc(\ket{\psi})\equiv   \Enc  \ket{\psi}\otimes \ket{0^{\Anc}} = \Enc_Q \big(\Enc_{\PMD} \otimes \mathbb{I}^{\Anc_{Q}}\big) \ket{\psi}\otimes \ket{0^{\Anc}}
\end{equation}

We remark that $|\Anc| = n-k$, and $\Anc_Q = n-m$. See \cref{fig:setup} for an illustration. 

\subsubsection{Decoding}

Our decoding algorithm $\Dec$ is composed of two phases, a list-decoding step in \Cref{alg:algorithm1} to obtain a discrete set of correction operators, followed by a coherent decoding phase in \Cref{alg:algorithm2} to correct these errors in superposition. While for simplicity we describe the algorithm in the context of non-adaptive adversaries, the algorithm for adaptive adversaries is exactly the same.

The first phase follows the ideas in \cite{Bergamaschi2022ApproachingTQ} adapted to the erasure list-decoding of \Cref{section:QLDE}. Upon receiving a reduced density matrix $\rho_T$ of a code-state $\Enc (\psi)$ for some $T\subset [n]$, we first measure the syndrome $s$ of $\rho_T \otimes 0^{n-|T|}$ using a description of the stabilizers of $Q$. If $|T|\geq (1-\delta)\cdot n$ and $Q$ is $(\delta, L)$-QLDE, then the syndrome measurement collapses the state into a superposition of at most $L$ logically-distinct errors acting on the $n$-qudit code register, and we can efficiently obtain a list $E_1\cdots E_L$ of candidate Pauli corrections (\Cref{claim:syndromecollapse}). 

\begin{algorithm}[h]
    \setstretch{1.35}
    \caption{$\Dec$ }
    \label{alg:algorithm1}
    \KwInput{$T\subset [n]$, and a reduced density matrix of a code-state $\rho_T = \text{Tr}_{[n]\setminus T}[\Enc (\psi\otimes 0^{\Anc})]$}

    \KwOutput{A $k$ qudit state $\psi'\approx \psi$ close to the original message}

    \begin{algorithmic}[1]

    \State Prepare the state $\rho_T \otimes 0^{n-|T|}$ on an $n$ qudit register $C$, and measure the syndrome $s$ of $Q$.\

    \State Apply erasure list-decodability to the syndrome $s$ and $T$ to obtain a list $E_1, \dots, E_L$ of potential errors.\

    \State Coherently Correct using \Cref{alg:algorithm2}, producing a state supported on the register $C$ and $L$ extra qubits. 

    \State Output the first $k$ qudits of $C$. 
            
    \end{algorithmic}

\end{algorithm}

In the second phase (\Cref{alg:algorithm2}), we filter through each list element coherently, attempting each correction operator sequentially and relying on the $\PMD$ authentication unitary $\Auth$ to ``catch" invalid corrections. Recall that $\Auth$ coherently measures the projection onto the $\PMD$ code-space, and reverts the encoding if succesful (\Cref{claim:pmdapproximation}).

In slightly more detail, we begin by preparing $L$ extra ``flag" qubit registers $F_1\cdots F_L$ initialized to $\ket{0^L}$, which will be used to coherently keep track of whether the state has been decoded already. First, we apply $\Enc_Q^\dagger E_1^\dagger$ to the code register $C$ to revert the state back into a corrupted codeword of a $\PMD$, and we proceed by applying $\Auth$ on the $\PMD$ register and the first Flag qubit $F_1$. Conditioned on $F_1$ failing, we apply $\Enc_Q^\dagger\big( E_{i+1}^\dagger E_i \big)\Enc_Q$ to attempt the next correction. We repeat this process iteratively over $i\in [2, L]$, conditioning on the absense of previous success. 

\begin{algorithm}[h]
    \setstretch{1.35}
    \caption{$U$: Coherently Correcting using $\PMD$s and QLDs }
    \label{alg:algorithm2}
    \KwInput{A list $\{E_i\}_{i\in [L]}$ of Pauli operators, and an $n$ qudit register $C$ with the state $\ket{\phi}$}
    
    \KwOutput{The pure state $U\ket{
    \phi}_C\ket{0}^L$ supported on $n$ qudits and $L$ qubits.}

    \begin{algorithmic}[1]
    
    \State Prepare the state $\ket{\phi}_C \otimes \ket{0^{\otimes L}}_F$, by appending $L$ extra ``Flag" qubit registers $F_1\cdots F_L$. 
    
     \State Apply $(\Enc_Q^\dagger E_1^\dagger)_C\otimes \mathbb{I}_F$ to revert the state to a corrupted $\PMD$ code-state, followed by $(\Auth_{\PMD, F_1}\otimes\mathbb{I}_{
     \Anc_Q}\otimes\mathbb{I}_{F_{[n]\setminus \{1\}}})$ and authenticate it and write the output on $F_1$

     \State Revert the candidate correction, controlled on the $F_1$
     \begin{equation}
            U_{1}=\Enc_Q^\dagger\big( E_{2}^\dagger E_1 \big)\Enc_Q \otimes \ket{0}\bra{0}_{F_{1}} + \mathbb{I}_C \otimes \ket{1}\bra{1}_{F_{1}} 
        \end{equation}

    \State For $i\in [2, L]$

\State \Indp Apply $\Auth$ to the $\PMD$ register and the $i$th flag, controlled on the $i-1$st flag, 
\begin{equation}
       A_i =  \Auth_{C, F_i}\otimes \ket{0}\bra{0}_{F_{i-1}} + \mathbb{I}_C\otimes X_{F_i} \otimes \ket{1}\bra{1}_{F_{i-1}}
\end{equation}

        \State Attempt the next correction, controlled on the $i$th flag, 
        \begin{equation}
            U_i = \Enc_Q^\dagger\big( E_{i+1}^\dagger E_i \big)\Enc_Q \otimes \ket{0}\bra{0}_{F_{i}} + \mathbb{I}_C \otimes \ket{1}\bra{1}_{F_{i}} 
        \end{equation}

    \end{algorithmic}

\end{algorithm}

We show that \Cref{alg:algorithm2} succeeds in coherently producing the message $\psi$, by relying on the guarantee that the $\PMD$ approximately preserves the corrupted code-state when the authentication fails. In this manner, we are able to sequentially try elements of the list until one succeeds, without the error accumulating excessively.

\begin{lemma}\label{lemma:aqeccsfrompmds}
    Let $\PMD$ be a $(m, k, \epsilon)_q$-$\PMD$ and $Q$ be a $[[n, m]]_q$ stabilizer code which is $(\delta, L)$-QLDE, both defined over the $q$-ary Pauli basis $\mathcal{P}_q$. Then $(\Enc, \Dec)$ define a quantum code encoding $k$ qudits into $n$ qudits, which $(\delta, 3\cdot \epsilon^{1/2}\cdot L^{3/4})$ approximately corrects erasures. 
\end{lemma}

We can now instantiate the lemma above with our constructions of list decodable codes and PMD codes, to arrive at \Cref{theorem:main}.

\begin{corollary}\label{corollary:aqecc-parameters}
For all $0<r<1$, every sufficiently large $n\in \mathbb{N}$, and $\gamma \geq (\log n)^{-1}$, there exists a Monte Carlo construction of a family of $((n, r\cdot n))_q$ quantum codes which $(\frac{1}{2}(1-r-\gamma), 2^{-\Omega(\gamma \cdot n)})$ approximately corrects erasures in time $n^{O(1)}$. The construction succeeds with probability $1-2^{-\Omega(\gamma \cdot n)}$.
\end{corollary}

\begin{proof}

[of \Cref{corollary:aqecc-parameters} and \Cref{theorem:main}] If the PMD has rate $r_\PMD$ and $Q$ has rate $r_Q$, then $R = r_\PMD\cdot r_Q$. We use the family of PMDs from \Cref{theorem:results-pmdexplicit}, such that $\epsilon_\PMD = n^{O(1)}\cdot 2^{-\Omega(\lambda)}$, where $\lambda = r_Q\cdot (1-r_\PMD)\cdot n$. 
    
    We use the family of $(\frac{1}{2}(1-r_{Q}-\gamma_Q), 2^{O(1/\gamma_Q)})$ erasure list decodable codes from \Cref{corollary:randomcss}. We pick a $\gamma'$ s.t. (1) $r+\gamma \geq R+\gamma' = r_Q+\gamma_Q$, to match the desired erasure decoding radius, and thereby $\lambda = (\gamma'-\gamma_Q)\cdot n = \gamma_Q\cdot n$ if $\gamma_Q=\gamma'/2$. The theorem statement is vacuous unless $R<1-\gamma'$, and thereby $r_Q = R+\gamma'/2$ is well defined. 

    The corresponding $r_\PMD = 1 - \frac{\gamma_Q}{\gamma_Q +  R}$, is only well defined in \Cref{theorem:results-pmdexplicit} if (2) the fraction is rational and $< 1/3$. If the desired rate $r\geq \gamma/2$, then we pick $R = r$ and $\gamma_Q = \gamma/50$. If the desired rate $r\leq \gamma/2,$ then we pick $R=\gamma/2$, and $\gamma_Q = \gamma/50$. In both situations, we satisfy conditions (1, 2), have rate $\geq r$, decoding radius $\geq \frac{1}{2}(1-r-\gamma)$, error $2^{O(1/\gamma_Q)}\cdot 2^{-\Omega(\gamma_Q n)} = 2^{-\Omega(\gamma \cdot n)}$, and runtime $2^{O(1/\gamma_Q)}\cdot n^{O(1)} = n^{O(1)}$. The construction of $Q$ succeeds with probability $\geq 1- n^{O(1)}\cdot 2^{-\Omega(\gamma_Q\cdot n)}= 1-2^{-\Omega(\gamma\cdot n)}$.
\end{proof}

\subsection{Analysis}
\label{section:aqecc-analysis}

Recall that \Cref{alg:algorithm1} performs a syndrome measurement, collapsing the $n$ qudit corrupted code-state $\mathcal{A}\circ \Enc(\psi)$ into a superposition of Pauli errors $\ket{\phi} = \sum_{i\in [L]}a_i \cdot E_i\Enc(\ket{\psi})$ of the same syndrome. Since $Q$ is $(\delta, L)$-QLDE, there are at most $L$ logically-distinct elements in the superposition (\Cref{claim:syndromecollapse}). Thus, it suffices to show that the unitary $U$ defined in \Cref{alg:algorithm2} approximately recovers from these superpositions, producing a state close to a product state with the original message $\ket{\psi}$.

We first claim that $U$ approximately recovers from single Pauli errors in \Cref{claim:singlepaulirecovery}, corresponding essentially to our ``Sparse Pauli Channel" \Cref{lemma:results-sparsepauli} in the overview. This is the most technical part of the analysis, and we defer its proof to the bottom of this section. Then, in \Cref{claim:multipaulirecovery} we show how to boot-strap this claim into recovering from superpositions of at most $L$ errors, which all but immediately gives us  \Cref{lemma:aqeccsfrompmds} as a corollary. 

\begin{claim}[Single Pauli Errors]\label{claim:singlepaulirecovery}
    Let $\ket{\phi_i} = E_i\Enc(\ket{\psi})$, where $E_i$ is logically equivalent to the $i$th list element in \Cref{alg:algorithm1}. Then $U$ approximately recovers the encoded message, i.e., $\forall \psi$:
    \begin{equation}
        \big\| \ket{\psi}\otimes \ket{\Aux_i} - U \cdot (\ket{\phi_i} \otimes \ket{0^L})\big\|_2  \leq 2\cdot L\cdot \epsilon
\end{equation}
    For some choice of $n-k$ qudit and $L$ qubit ancilla $\ket{\Aux_i}$. Moreover, $\bra{\Aux_i}\ket{\Aux_j}=\delta_{ij}$.
\end{claim}

The intuition behind the proof of the single Pauli error case of \Cref{claim:singlepaulirecovery} lies in considering an ``ideal'' $\PMD$, with error $\epsilon = 0$. In which case, all the authentication steps don't change the state, until a ``correct" list element $E_j^\dagger$ is accepted by the $\PMD$, returning the original message $\psi$.

To lift our Sparse Pauli Channel correction to correcting from \textit{superpositions} of errors, we fundamentally leverage the fact that corrupted code-states of PMD codes are near-orthogonal. We show that this limits the quantum interference between distinct Pauli errors, and enables us to correct from sparse superpositions with only a slight degrade in accuracy:

\begin{claim} [Sparse Superpositions of Pauli Errors] \label{claim:multipaulirecovery}
    If  $\ket{\phi} = \sum_{i\in [L]}a_i E_i\Enc(\ket{\psi})=\sum_{i\in [L]} a_i \ket{\phi_i}$ 
    where $\bra{\phi} \ket{\phi}=1$, then $U$ approximately recovers the encoded message $\psi$, i.e. 
        \begin{equation}
        \big\| \psi\otimes \Aux - U \circ (\phi \otimes 0^L)\big\|_1  \leq 3\cdot \epsilon^{1/2}\cdot L^{3/4}
    \end{equation}

For some choice of $n-k$ qudit and $L$ qubit state $\ket{\Aux}$.
\end{claim}

This concludes the proof that $\Dec$ recovers the message $\ket{\psi}$ from adversarial erasure channels (\Cref{definition:aqec}), thus $(\Enc, \Dec)$ forms an erasure AQECC as in \Cref{lemma:aqeccsfrompmds}.

\subsection{The Proofs}
\label{subsection:aqecc-proofs}

We dedicate the rest of this section to proofs of the claims above. 

\begin{proof} 

[of \Cref{claim:singlepaulirecovery}] Since $Q$ is $(\delta, L)$-QLDE, the $i$th list element $E$ corrects $\ket{\phi_i}$, i.e., $E_i^\dagger\ket{\phi_i} = \Enc(\psi)$. If the previous authentication steps didn't corrupt the state, then the authentication should correctly accept on the $i$th iteration of \Cref{alg:algorithm2}: 
\begin{equation}
    \Auth \bigg(\Enc_Q^{\dagger} E^\dagger_i E_i \Enc(\ket{\psi}) \bigg) \ket{0}_{F_i} =  \Auth \bigg( \Enc_{\PMD}(\ket{\psi}) \bigg) \ket{0}_{F_i} = \ket{\psi}\otimes \ket{0^{\Anc}} \otimes \ket{1}_{F_i}
\end{equation}

During the previous iterations $t\in (1\cdots i-1)$, we sequentially apply the operators $E_{i+1}^\dagger E_i$ to change between code-states of $Q$. We emphasize that $E_{i+1}^\dagger E_i$ is a logical operator of $Q$, and thus $E_{ij} = \Enc^{\dagger}_Q E_j^\dagger E_i \Enc_Q$ is a $n$-qudit Pauli operator, with non-trivial support only on the $m$-qudit message register corresponding to the $\PMD$.

Consider the pure state $\ket{D_t}$ defined after the $t$-th iteration of \Cref{alg:algorithm2}, and assume $1\leq t < i$. Let the ideal pure state $\ket{v_t}$ be
\begin{equation}
    \ket{v_t} = \bigg(\Enc_Q^{\dagger} E_t^\dagger E_i \Enc(\ket{\psi}) \bigg) \ket{00\cdots 0}_{F} 
\end{equation}

We prove by induction that $\|\ket{D_t}-\ket{v_t}\|_2 \leq 2\cdot t\cdot \epsilon$ for $1\leq t < i$. We note that $t=1$ corresponds to \Cref{claim:pmdapproximation}. For $t>1$, we observe the recurrences
\begin{equation}
    \ket{D_t} = U_tA_t \ket{D_{t-1}}, \text{ and } \ket{v_t} = U_t \ket{v_{t-1}}
\end{equation}

Consider the decomposition $\ket{D_t} = U_t A_t  \ket{v_{t-1}} + U_t A_t (\ket{D_{t-1}}- \ket{v_{t-1}})$. Note that by \Cref{claim:pmdapproximation},
\begin{equation}
\|\ket{v_{t}}- U_t A_t\ket{v_{t-1}}\|_2 =
    \| \ket{v_{t-1}}- A_t\ket{v_{t-1}}\|_2 \leq 2\cdot \epsilon
\end{equation}
    
and that by the inductive hypothesis, we have $\|\ket{D_{t-1}}-\ket{v_{t-1}}\|_2 \leq 2\cdot (t-1)\cdot \epsilon$. Thus,
\begin{equation}
    \|\ket{D_t}-\ket{v_t}\|_2 \leq \|\ket{D_{t-1}}-\ket{v_{t-1}}\|_2+ \|U_t A_t\ket{v_{t-1}}-\ket{v_{t}}\|_2 \leq 2\cdot t\cdot \epsilon
\end{equation}

After succeeding at the $i$th attempt, the remaining gates $\prod_{k=i+1}^L (A_k U_k)$ in \Cref{alg:algorithm2} only increment the counter $F$ and don't act on the code $C$. Thus, the output state is $\leq 2\cdot L\cdot \epsilon$ close to the pure state
\begin{equation}
    \prod_{k=i+1}^L (A_k U_k)A_{i} \ket{v_{i}} = \ket{\psi} \otimes \ket{0^{\Anc}} \otimes \ket{0_10_2\cdots 0_{i-1}1_{i} 1_{i+1}\cdots 1_L}_{F}
\end{equation}

Orthogonality follows from the fact that no two logically distinct operators $i\neq j$ are accepted by the $\PMD$ at the same iteration.
\end{proof}

It remains now to show the proof of \Cref{claim:multipaulirecovery} from \Cref{claim:singlepaulirecovery}.

\begin{proof} 

[of \Cref{claim:multipaulirecovery}] We first show that the corrupted code-states are approximately orthogonal, $|\bra{\phi_i}\ket{\phi_j}| \leq \epsilon$, which via the normalization constraint $|\bra{\phi}\ket{\phi}|=1$, tells us $\sum_i |a_i|^2\approx 1$. Then we show \Cref{alg:algorithm2} produces a state close to $\psi\otimes \Aux$, where $\Aux$ is the superposition of aux states, $\ket{\Aux} \propto \sum_{i\in [L]} a_i \ket{\Aux_i}$.

Recall from \Cref{claim:singlepaulirecovery} that $\Enc_Q E_i^\dagger E_j\Enc_Q = E_{ij}$ is some Pauli operator acting on the $\PMD$ register. Thus, using \Cref{def:PMD2},
\begin{equation}
    |\bra{\phi_i}\ket{\phi_j}| = |\bra{0^{\Anc}}\bra{\psi}\Enc^\dagger E_i^\dagger E_j \Enc \ket{\psi} \ket{0^{\Anc}}| \leq \| \Pi E\Pi\|_\infty\leq \epsilon
\end{equation}
Let $\ket{aux} = N^{-1/2}\sum_{i\in [L]} a_i \ket{aux}_i$, where the orthogonality of the $\ket{aux}_i$ tells us $N = \sum_i |a_i|^2$. Using the normalization constraint on $\ket{\phi}$, we obtain a bound on $N$:
\begin{equation}
    1 = \bra{\phi}\ket{\phi} = \sum_i |a_i|^2 + \sum_{i\neq j}a_i^*a_j \bra{\phi_i}\ket{\phi_j}\geq \sum_i |a_i|^2 - \epsilon\cdot \sum_{i\neq j}|a_i^*a_j|\geq (1-\epsilon\cdot L)\cdot \sum_i |a_i|^2,
\end{equation}
via the Cauchy-Schwartz inequality. Thereby $N = \sum_i |a_i|^2\leq 1/(1-\epsilon\cdot L)$.

 Finally, we consider the desired inner product:
    \begin{gather}
       N^{1/2}\cdot \big|\big(\bra{\psi} \otimes \bra{\Aux}\big) U \ket{\phi} \otimes \ket{0^L}\big| = \big|\sum a_j^*  \bra{\psi} \otimes \bra{\Aux_j} U \ket{\phi} \otimes \ket{0^L}\big| \\ \geq \braket{\phi} - \sum |a_j|\cdot \bigg|  \bigg(\bra{\psi} \otimes \bra{\Aux_j} U - \bra{\phi_j}\otimes\bra{0^L}\bigg) \ket{\phi} \otimes \ket{0^L} \bigg| \geq \\ \geq 1 - \sum_j |a_j|\cdot \big\| \ket{\psi}\otimes \ket{\Aux_j} - U \cdot (\ket{\phi_j} \otimes \ket{0^L})\big\|_2 \geq 1-2 \cdot L^{3/2} \cdot \epsilon \cdot (\sum_i |a_i|^2)^{1/2}
       \end{gather}

    Where we used the triangle inequality, the Cauchy-Schwartz inequality twice, and then \Cref{claim:singlepaulirecovery}. Thus, 
    
    \begin{equation}
         \big|\big(\bra{\psi} \otimes \bra{\Aux}\big) U \ket{\phi} \otimes \ket{0^L}\big| \geq  (1-L\cdot \epsilon)^{1/2} - 2 \cdot L^{3/2} \cdot \epsilon  \geq 1 - 3 \cdot L^{3/2} \cdot \epsilon 
    \end{equation}

    The relation between inner product and trace distance of pure states gives us the desired bound. 
\end{proof}

\newpage
\section{Tamper Detection and Non-Malleability for Qubit-wise Channels}
\label{section:auth}

Recall \Cref{definition:qa} of quantum tamper detection codes for qubit-wise channels:

\begin{definition} 
    A pair of quantum channels $(\Enc, \Dec)$ is an $\epsilon$-secure quantum tamper detection code for qubit-wise channels if they satisfy
\begin{enumerate}
    \item (Correctness) In the absence of tampering, for all messages $\psi_{MR}$: $( \Dec\circ \Enc\otimes \mathbb{I}_R )(\psi_{MR}) = \Acc\otimes \psi_{MR}$.
    \item (Tamper Detection) For every set of $n$ single qubit channels $\Lambda_1, \cdots, \Lambda_n$, there exists a constant $p_\Lambda\in [0, 1]$ such that, $\forall \psi_{MR}$
    \begin{equation}
        \bigg( \Dec\circ \big(\otimes_i^n \Lambda_i \big)\circ  \Enc\otimes \mathbb{I}_R \bigg)(\psi_{MR}) \approx_\epsilon p_\Lambda\cdot \Acc\otimes \psi  + (1-p_\Lambda)\cdot  \bot\otimes \psi_R  
    \end{equation}

    Where the error is measured in trace distance, and $\bot$ on $M$ indicates the message is rejected. 
\end{enumerate}
\end{definition}

Our main result in this section is \Cref{theorem:auth-results}, restated below:

\begin{theorem}
    For every sufficiently large $N\in \mathbb{N}$, there exists a Monte Carlo construction of an efficient quantum tamper detection code for qubit-wise channels of blocklength $N$, rate $1-O(\log^{-1} N)$, and error $\epsilon \leq 2^{-\Tilde{\Omega}(N)}$. The construction succeeds with probability $1-2^{-\Tilde{\Omega}(N)}$.
\end{theorem}

To build up to our main result, we present in \Cref{section:auth-rate-third}
 a self-contained proof of a simpler construction approaching rate $1/3$: 

 \begin{theorem}\label{theorem:one-third-auth}
For every sufficiently large $N\in \mathbb{N}$, there exists an explicit and efficient quantum tamper detection code for qubit-wise channels of blocklength $N$, rate $\frac{1}{3}(1-O(\log^{-1} N))$, and error $\epsilon = 2^{-\Tilde{\Omega}(N^{1/7})}$. 
\end{theorem}

Before moving on, we refer the reader to \cref{section:prelim} for the relevant background on classical non-malleable codes. 


\subsection{The Rate \texorpdfstring{$1/3 - o(1)$}{1/3-o(1)} Construction}
\label{section:auth-rate-third}

\subsubsection{The Code Construction}

\textbf{Encoding} Our code construction combines classical non malleability and approximate quantum error correction. We first encode the code-state $\psi$ into the composition $\Tilde{Q}$ of a PMD code, and a stabilizer code $Q$ of block-length $N_Q$. The resulting state is one-time-padded with a random $N_Q$ qubit Pauli $P^s$, indexed by a bit-string $s$ of length $2\cdot N_Q$. Finally, $s$ is encoded with a classical non-malleable code into $N_{\NM}$ bits, using the key-encapsulation technique. 
\begin{gather}
    \Enc(\psi) = \mathbb{E}_s \Enc_{\NM}(s)\otimes  P^s \Enc_{\Tilde{Q}}(\psi)(P^s)^\dagger, \text{ where } \Enc_{\Tilde{Q}}= \Enc_Q\circ\Enc_{\PMD}
\end{gather}

\textbf{Ingredients} To formally describe the ingredients, we combine

\begin{enumerate}
    \item The explicit non-malleable code of block-length $N_{\NM}$, rate $1-\gamma_{\NM}\geq 1 - 1/\log N_\NM$, and error $\epsilon_\NM = 2^{-\Tilde{\Omega}(N_\NM^{1/7})}$, of \cref{theorem:cg13} \cite{Cheraghchi2013NonmalleableCA}.
    \item An $[[N_Q, (1-\gamma_Q)\cdot N_Q, d]]_2$ stabilizer code with near-linear distance $d = \Omega(\gamma_Q\cdot N_Q/\log N_Q)$, such as that of \cref{theorem:bqrs} from \cite{Grassl_1999}.
    \item An $(k/(1-\gamma_\PMD), k, \epsilon_{\PMD})$-$\PMD$, $\epsilon_\PMD \leq k\cdot 2^{- \gamma_\PMD\cdot k/4}$ such as that of \Cref{thm:pmds}.
\end{enumerate}

Note that the total block-length $N = N_\NM+N_Q = N_Q\cdot (1 + 2/(1-\gamma_{\NM}))$, and thereby the rate of the resulting construction is 

\begin{equation}
    R = \frac{1}{(1 + 2/(1-\gamma_{\NM}))} \frac{k}{N_Q} = \frac{1-\gamma_{\NM}}{3-\gamma_\NM}\cdot (1-\gamma_\PMD)\cdot (1-\gamma_Q)=\frac{1}{3} - O(\log^{-1} N )
\end{equation}

So long as we pick $\gamma_{\NM}, \gamma_{\PMD}, \gamma_Q  = \Theta(\log^{-1} N )$.

\textbf{Decoding} We describe our decoding algorithm $\Dec$ as follows. The first step is to decode the classical information in $C$ using $\Dec_\NM$, obtaining a key $\Tilde{s}$. In the second step, we attempt to decrypt the error correcting code by applying $(P^{\Tilde{s}})^\dagger$, and proceed by measuring the syndrome of $Q$. In the third, if the syndrome measurement succeeds, we decode $Q$ and project onto the PMD code-space. If all these steps accept, we output the message register.

\subsubsection{Analysis of the Rate $1/3$ construction}

We divide the proof of security into 3 claims. First, in \Cref{claim:classical-nm-third}, we reason that the classical non-malleability guarantees that Bob decodes the classical key to a convex combination over the original random key, or an uncorrelated key $\Tilde{s}$ (or rejects). Then, we analyze each of these two cases separately.

\begin{claim}\label{claim:classical-nm-third}
    For all messages $s$, the distribution over the classical outputs of the $\NM$ decoding algorithm is close to a convex combination over the original message $s$, rejection $\bot, $ and an uncorrelated distribution $\mathcal{D}^\Lambda$.
    \begin{equation}
        \|\Dec_\NM \circ \big(\otimes_i \Lambda_i\big)\circ \Enc_\NM(s) - p_{\Lambda, s}\|_1 \leq \epsilon_{\NM}
    \end{equation}

That is, $p_{\Lambda, s}(\Tilde{s}) = p_{\Lambda, \text{Same}}\cdot  \delta_{s, \Tilde{s}} + p_{\Lambda, \bot}\cdot \delta_{\bot, \Tilde{s}} + (1-p_{\Lambda, \text{Same}}-p_{\Lambda, \bot})\cdot \mathcal{D}^\Lambda(\Tilde{s})$, where $p_{\Lambda, \text{Same}}, p_{\Lambda, \bot}$ are positive constants which only depend on $\Lambda$.
\end{claim}

If Bob receives the true secret key $s$, then one can show that on average over random $s$, he receives a code-state of $\Tilde{Q}$ corrupted by a Pauli channel, which the inner PMD code can detect: 

\begin{claim}\label{claim:third-auth-recovered}
    If Bob recovers the key, then for all (possibly entangled) messages $\psi = \psi_{MR}$,
    \begin{equation}
        \bigg\|\mathbb{E}_s \Dec_{\Tilde{Q}}\circ (P^s)^\dagger \circ \Lambda\circ P^s \circ \Enc_{\Tilde{Q}}(\psi) - \bigg(p_{\Lambda, \text{a}}\cdot \psi\otimes \Acc + \bot\cdot (1-p_{\Lambda, \text{a}})\bigg) \bigg\|_1 \leq \epsilon_\PMD^2 \leq 2^{-\Omega(k/\log k)}
    \end{equation}

    Where we denote rejection as $\bot = \Rej\otimes \sigma_{MR}$ (for some density matrix $\sigma_{MR}$), and $p_{\Lambda, \text{a}}$ is a real positive constant.
\end{claim}

Conversely, if Bob receives the uncorrelated key, then on average over the random Pauli OTP Bob's syndrome measurement on $\Tilde{Q}$ should fail with high probability: 

\begin{claim}\label{claim:third-auth-not-recovered}
    If Bob receives an uncorrelated key $\Tilde{s}$, then for all messages $\psi$,
    \begin{equation}
        \bigg\|\mathbb{E}_s \Dec_{Q}\circ (P^{\Tilde{s}})^\dagger \circ \Lambda\circ P^s \circ \Enc_{\Tilde{Q}}(\psi) - \bot \bigg\|_1 \leq 2^{-\Omega(d^2/N_Q)} \leq 2^{-\Omega(-k/\log^4 k)}
    \end{equation}
\end{claim}

We can now tie these ideas together and prove $(\Enc, \Dec)$ forms a tamper detection code.

\begin{proof}

    [of \Cref{theorem:one-third-auth}] After Bob receives the corrupted code-state $\big(\otimes_i\Lambda_i\big) \circ \Enc(\psi)$, he measures and decodes the classical register obtaining a key $\Tilde{s}$, and attempts to revert the OTP $(P^{\Tilde{s}})^\dagger$. From \Cref{claim:classical-nm-third}, after tracing out the `key' register, Bob has in his possession a mixed state $\rho_{\text{Bob}}$ which is $\epsilon_\NM$ close to $\rho^{\text{ideal}}_{\text{Bob}}$, a mixture  over whether the key is recovered, uncorrelated, or rejected:
    \begin{gather}
       \rho^{\text{ideal}}_{\text{Bob}} = p_{\Lambda, \text{Same}} \cdot \mathbb{E}_s \bigg((P^s)^\dagger \circ \Lambda\circ P^s \circ \Enc_{\Tilde{Q}}(\psi)\bigg) + \\ + p_{\Lambda, \bot}\cdot \bot + (1-p_{\Lambda, \text{Same}}-p_{\Lambda, \bot}) \cdot \mathbb{E}_{\Tilde{s}\leftarrow \mathcal{D}^\Lambda}\mathbb{E}_s \bigg((P^{\Tilde{s}})^\dagger \circ \Lambda\circ P^s \circ \Enc_{\Tilde{Q}}(\psi)\bigg)
    \end{gather}
    By the triangle inequality, \Cref{claim:third-auth-recovered} and \Cref{claim:third-auth-not-recovered}, $\Dec_{\Tilde{Q}}(\rho^{\text{ideal}}_{\text{Bob}})$ is $\max(\epsilon_\PMD^2, 2^{-\Omega(-d^2/N_Q)})$ close to a mixture over the original $\psi$, and rejection $\bot$. This immediately implies this construction is an tamper detection code with error $\epsilon_{\NM} + \max(\epsilon_\PMD^2, 2^{-\Omega(-d^2/N_Q)})\leq 2^{-\Tilde{\Omega}(N^{1/7})}$
\end{proof}

\textbf{The Proofs (of the Rate $1/3$ construction)}

\begin{proof}

    [of \Cref{claim:classical-nm-third}] This claim is a well-known consequence of the definition of non-malleability. After applying the channels $\Lambda_i$ and measuring the output, each channel simplifies to to a randomized function $f_i$. From \Cref{definition:nm-codes}, there exists a distribution $\mathcal{D}_f$ which determines the distribution over decoding outputs. Crucially, $\mathcal{D}_f$ is a convex combination over the original message, rejection, and an arbitrary message, and the coefficients of the convex combination do not depend on the message.
\end{proof}

\begin{proof}

[of \Cref{claim:third-auth-recovered}] The Pauli Twirl (\cref{fact:1design}) guarantees the effective channel is a Pauli channel. 
    \begin{equation}
        \mathbb{E}_s (P^s)^\dagger \circ \Lambda\circ P^s \circ \Enc_{\Tilde{Q}}(\psi)  = \sum_E |c_E|^2\cdot E \Enc_{\Tilde{Q}}(\psi) E^\dagger
    \end{equation}
    The syndrome measurement on the stabilizer code $Q$, assuming it accepts, outputs 
    \begin{equation}
        \Pi_Q\circ \bigg(\mathbb{E}_s (P^s)^\dagger \circ \Lambda\circ P^s \circ \Enc_{\Tilde{Q}}(\psi)\bigg)  = \sum_{E\in N(Q)} |c_E|^2\cdot E \Enc_{\Tilde{Q}}(\psi) E^\dagger
    \end{equation}

    After reverting the stabilizer code encoding $\Enc_Q^\dagger$ and projecting onto the $\PMD$ code-space, the operators $E$ in $S(Q)$ preserve the state, and those in $N(Q)-S(Q)$ are detected by the $\PMD$ with probability $1-\epsilon^2_\PMD$:
    \begin{gather}
       \bigg\| \Enc_\PMD^\dagger\circ \Pi_{\PMD}\circ \Enc_Q^\dagger \circ \bigg( \sum_{E\in N(Q)} |c_E|^2 \cdot E \Enc_{\Tilde{Q}}(\psi) E^\dagger\bigg) - \psi\cdot \bigg(\sum_{E\in S(Q)} |c_E|^2\bigg) \bigg\|_1\leq \\ \leq \epsilon^2_\PMD\cdot \sum_{E\in N(Q)} |c_E|^2 \leq \epsilon^2_\PMD
    \end{gather}

    Where we used the fact that $\sum_{E} |c_E|^2=1$ since $\Lambda$ is CPTP. The observation that $\epsilon_\PMD \leq 2^{-\Omega(k/\log k)}$ concludes the claim.   
\end{proof}

\begin{proof}

[of \Cref{claim:third-auth-not-recovered}] If the classical key is not recovered, the decoder ``receives" a key $\Tilde{s}$ which is uncorrelated from $s$ by the non-malleability guarantee. The density matrix in Bob's possession, on average over the secret key but before decoding the quantum component of the code, is a product state which doesn't depend on $\psi$ (by \cref{fact:1design}):
\begin{equation}
    \mathbb{E}_s \Lambda\circ P^s \circ \Enc_{\Tilde{Q}}(\psi) = \Lambda\circ \bigg(\mathbb{E}_s P^s \cdot \Enc_{\Tilde{Q}}(\psi) \cdot (P^s )^\dagger\bigg) = \Lambda(\mathbb{I}/2^{N_Q}) = \otimes_i^{N_Q} \Lambda_i(\mathbb{I}/2)
\end{equation}

However, code-states of stabilizer codes should be highly entangled, and very far from product states. Indeed, \cref{lemma:clb} by \cite{Anshu2020CircuitLB} implies that the syndrome measurement on $Q$ must fail (that is, output a non-zero syndrome) with high probability if the key is not recovered.
\begin{equation}
    \Tr[\Pi_Q\cdot  \otimes_i^{N_Q} \Lambda_i(\mathbb{I}/2)] \leq 2\cdot 2^{-d^2/(2^{12}\cdot N_Q)}
\end{equation}

The observation that $N_Q \leq k(1+\Theta(1/\log k))$, $d = \Omega(k/\log^2 k)$ concludes the claim.

\end{proof}

\subsection{The Rate \texorpdfstring{$1-o(1)$}{1-o(1)} Construction}

\subsubsection{Overview}

The code construction is comprised of a classical non-malleable code and multi-layer concatenated quantum code. In \cref{section:auth-ingredients-rate-1}, we formally describe the components in our code construction. In \Cref{section:encoding_rate_1}, we describe our encoding map (\Cref{alg:enc_qnm}), although we recommend the reader to skim \Cref{fig:concat} for an intuitive description of the algorithm first. We dedicate \Cref{section:analysis-rate1} to the key claims behind the analysis of our construction, which are all mostly independent of the parameter choices. Finally, in \Cref{section:proofs-rate1}, we present the proofs of the claims in the analysis. 

The quantum half of the construction essentially consists of the concatenation of an outer code $Q_o$, and an inner code $\Tilde{Q}_{in}$. $\Tilde{Q}_{in}$ is the composition of a $\PMD$ code and a stabilizer code of large ``pure" distance (the minimum weight of both normalizers and stabilizers). We will show that this composition has very robust error detection properties, specifically against qubit-wise channels.

We use the classical non-malleable code to ``hide" a short classical secret key, which in turn is used to sample a $t$-wise independent Pauli one-time-pad applied to the quantum code. At a high level, if $t>b$, then on each inner block the OTP is uniformly random, and an analysis similar to the uniformly random case will hold. 

\subsubsection{Ingredients}
\label{section:auth-ingredients-rate-1}

To formally describe the ingredients, we combine:
\begin{enumerate}
    \item A classical non-malleable code $(\Enc_{\NM}, \Dec_{\NM})$ of message length $N_\NM$, rate $r_{\NM} = \Theta(1)$, and error $\epsilon_{\NM}\leq 2^{-\Omega(N_\NM)}$, such as that of \cref{theorem:dpw10} by \cite{Dziembowski2010}. This construction is Monte-Carlo with failure probability $2^{-\Omega(N_\NM)}$.
    
    \item A $t$-wise independent function $G:\{0, 1\}^{|s|}\rightarrow \{0, 1\}^{2N_Q}$, of seed length $|s| = O(t\log N_Q)$.
    \item A $[[n_{out}, k, d_{out}]]_2$ stabilizer code $Q_{out}$ with near-linear distance and pure distance\footnote{Refer to \cref{def:distance} for definitions of distance and pure distance for stabilizer codes.} $d_{out} = \Omega((n_{out}-k)/\log n_{out})$, such as that of \cref{theorem:bqrs} from \cite{Grassl_1999}. In particular we pick rate $1-\frac{1}{\log k}$, s.t. $d_{out}\geq \frac{k}{2\cdot \log^2 k}$.

    \item An $(n_{\PMD}, k_{\PMD}, \epsilon_{\PMD})$-$\PMD$ $\Pi$ from \Cref{theorem:results-pmdexplicit}, where $n_{\PMD} = k_{\PMD}+\lambda$ and $\epsilon_{\PMD}\leq 10\cdot k_{\PMD}\cdot 2^{-\lambda/4}$. In particular, we pick $\lambda = k_{\PMD}/\log k_{\PMD}$.
    
    \item A $[[b, n_{\PMD}, d^*]]_2$ stabilizer code $Q_{in}$ with near-linear distance $d^* = \Omega((b-n_\PMD)/\log b)$, such as \ref{theorem:bqrs} from \cite{Grassl_1999}. In particular, we pick rate $1-1/\log b$, such that $d^* \geq \frac{b}{2\log^2 b}$.
\end{enumerate}

\subsubsection{Encoding and Parameter Choices}
\label{section:encoding_rate_1}

We describe our encoding channel in \Cref{alg:enc_qnm}.

\begin{algorithm}[h]
    \setstretch{1.35}
    \caption{$\Enc$}
    \label{alg:enc_qnm}
    \KwInput{The message state $\ket{\psi}$}

    \KwOutput{The encoded mixed state $\Enc(\psi)$}

    \begin{algorithmic}[1]
    
    \State Encode $\psi$ into $\Enc_{Q_o}(\psi)$, a $[[n_{out}, k, d_{out}]]_2$ QRS code $Q_o$.\

    \State Partition the $n_{out}$ qubits of $Q_o$ into $n=n_{out}/k_{\PMD}$ consecutive blocks of $k_{\PMD}$ qubits.\
    
    \State Encode each block $j\in [n]$ of $Q_o$ into the composition $\Tilde{Q}_j$ of an $(n_\PMD, k_\PMD, \epsilon_\PMD)_2$-$\PMD$ code, and a $[[b, n_{\PMD}, d^*]]$ stabilizer code.
    \begin{equation}
        \Enc_{\Tilde{Q}} = \big(\otimes_{j\in [n]}\Enc_{\Tilde{Q}_j}\big) \big(\Enc_{Q_o}(\psi)\big)
    \end{equation}
    
    \State Sample a secret key $s$, and apply a $b\cdot n$ qubit Pauli OTP $P^{G(s)}$ to $\Enc_{\Tilde{Q}}$. \

    \State Encode $s$ into $\Enc_{\NM}(s)$, a classical bit-wise non-malleable code.
    \begin{equation}
    \Enc(\psi) = \mathbb{E}_s \Enc_{\NM}(s)\otimes P^{G(s)}\cdot \Enc_{\Tilde{Q}} (\psi) \cdot (P^{G(s)})^\dagger
    \end{equation}
    \end{algorithmic}
\end{algorithm}

 Note that to adequate for the appropriate message length of the inner codes, in the concatenation step in \Cref{alg:enc_qnm} we group the qubits of $Q_{out}$ into ``blocks'' of size $k_\PMD$, which are then encoded into the PMD and $Q_{in}$.

 \begin{claim}\label{claim:parameter_choices}
    Let us fix the choice $k_\PMD = n_{out}/\log^c k$, for some positive integer $c>10$, and $t = b\cdot \log k$. Then, the rate of the construction is at least $1-\gamma$, for some $\gamma = \Theta(1/\log k)$.
\end{claim}

\begin{proof}
Note that $n_{out} = k / (1- 1/\log k)$, and thus $k_\PMD = n_{out}/\log^c k \leq k/\log^c k \cdot (1+2/\log k)$ for sufficiently large $k$. Also, the number of inner blocks is $n = n_{out}/k_\PMD = \log^c k$. 

Each inner block, after encoding into the PMD code and $Q_{in}$, has block-length $b \leq  n_{out}/\log^c k \cdot (1 - 1/\log k_\PMD)^{-2} \leq k/\log^c k \cdot (1+10/\log k)$. Note also that each inner code $Q_{in}$ has distance and pure distance $d \geq \frac{1}{4}\cdot k/\log^{c+2} k$; and the error of each inner PMD code is $\epsilon_{\PMD} \leq 2^{-k/(10\cdot\log^{c+1} k)}$. 

The overall quantum block-length is $N_Q = b\cdot n \leq k\cdot (1+10/\log k)$. Since we pick $t = b\cdot \log k$, the length of the classical seed is $|s| = O(k\cdot \log^{1-c} k)$; and the length of the non-malleable encoding $N_\NM = |s|/r_{\NM} = O(k\log^{1-c} k)$. The total length of the encoding is thus $N=N_Q+N_\NM \leq k\cdot (1+ 13/\log k + O(\log^{1-c} k))$.
\end{proof}

\subsubsection{Decoding}

The first step in our decoding algorithm $\Dec$ (\Cref{alg:dec_qnm}) is to decode the classical information in $C$ using $\Dec_\NM$, obtaining a key $\Tilde{s}$. This decoded key $\Tilde{s}$ is then used to revert the pseudorandom Pauli on the quantum register, introduced in \Cref{alg:enc_qnm}.

We proceed by using the inner codes to detect whether they have been tampered with. They do so with the channel $\Dec_{\Tilde{Q}_{in}}$ (\Cref{alg:dec_qnm_inner}), which measures the syndrome of the stabilizer code, and projects onto the PMD subspace. If all the inner codes accept, then we turn to the outer code. Naturally, if any of these steps rejects, we output $\bot$.

\begin{algorithm}[H]
    \setstretch{1.35}
    \caption{$\Dec$}
    \label{alg:dec_qnm}
    \KwInput{The qubit-wise corrupted code-state $\big(\otimes_{i\in N}\Lambda_i\big)\circ \Enc(\psi)$}

    \KwOutput{The mixed state $\rho_{out} \approx  p_{a}\cdot \psi\otimes \Acc+\bot\cdot (1-p_a)$, for some $p_a\in [0, 1]$}

    \begin{algorithmic}[1]    
    \State Decode the $\NM$ register using $\Dec_\NM$, producing a key $s$. If it rejects, output $\bot$.\

    \State Attempt to revert the pseudorandom Pauli pad by applying $(P^{G(s)})^\dagger$.\
    
    \State Iterate over each inner code $\Tilde{Q}_j$, $j\in [n]$, and apply $\Dec_{\Tilde{Q}_j}$. If any $j\in [n]$ inner codes rejects, output $\bot$. \

    \State Measure the syndrome of the outer code $Q_o$. If the syndrome is non-zero, $\bot$. If it accepts, revert $\Enc_{Q_o}^\dagger$.\
    \end{algorithmic}

\end{algorithm}

\begin{algorithm}[H]
    \setstretch{1.35}
    \caption{$\Dec_{\Tilde{Q}_{in}}$ on the $j$th inner code}
    \label{alg:dec_qnm_inner}
    \KwInput{A $b$-qubit register.}

    \KwOutput{A $k_\PMD$ qubit register. }

    \begin{algorithmic}[1]
    \State Perform a syndrome measurement of the stabilizer code $Q_j$. If the syndrome is non-zero, output $\bot$.\
    
    \State If the syndrome measurement accepts, revert the encoding unitary $\Enc_{Q_j}^\dagger$.\

    \State Project onto the $\PMD$ subspace $\Pi_j$ of the $j$-th inner code. \
    
    \State If it rejects, output $\bot$. If it accepts, revert $\Enc_{\PMD_j}^\dagger$, and output the $k_{\PMD}$ qubit message register.
    \end{algorithmic}

\end{algorithm}

\subsection{Analysis}
\label{section:analysis-rate1}

\textbf{Classical Non-Malleability} First, we point out again that the classical non-malleable code is used as a key-encapsulation mechanism, to establish a shared ``non-malleable" secret key between sender and receiver. In this manner, we can focus the rest of the analysis on whether the decoder faithfully receives the secret key, or if the key looks random to the decoder:

\begin{claim}[\Cref{claim:classical-nm-third}, restatement]\label{claim:classical-nm}
    For all messages $s$, the distribution over the classical outputs of the $\NM$ decoding algorithm is close to a convex combination over the original message $s$, rejection $\bot, $ and an uncorrelated distribution $\mathcal{D}^\Lambda$.
    \begin{equation}
        \|\Dec_\NM \circ \big(\otimes_i \Lambda_i\big)\circ \Enc_\NM(s) - p_{\Lambda, s}\| \leq \epsilon_{\NM}
    \end{equation}

That is, $p_{\Lambda, s}(\Tilde{s}) = p_{\Lambda, \text{Same}}\cdot  \delta_{s, \Tilde{s}} + p_{\Lambda, \bot}\cdot \delta_{\bot, \Tilde{s}} + (1-p_{\Lambda, \text{Same}}-p_{\Lambda, \bot})\cdot \mathcal{D}^\Lambda(\Tilde{s})$.
    
\end{claim}

\textbf{If the classical key is not recovered}, the decoder ``receives" a key $\Tilde{s}$ which is uncorrelated from $s$ (by the non-malleability guarantee). Let us consider the marginal density matrix on any subset $B\subset [n]$ of $\leq t/b$ inner codes, after the qubit-wise channel and after attempting to revert the Pauli pad $(P^{G(\Tilde{s})})^\dagger$. Via the $t$-wise independence of $G(s)$, these marginals are product states:
\begin{equation}
\text{Tr}_{\Bar{B}}\bigg[(P^{G(\Tilde{s})})^\dagger\circ \Lambda\bigg( \mathbb{E}_s\bigg[P^{G(s)}\cdot  \Enc_{\Tilde{Q}}(\psi)\cdot (P^{G(s)})^\dagger \bigg]\bigg)\bigg] = \otimes_{i\in B}  \Lambda_{i, \Tilde{s}}(\mathbb{I}/2) 
\end{equation}

where $\Lambda_{i, \Tilde{s}} = (P^{G(\Tilde{s})})^\dagger_i\circ \Lambda_{i}$. However, product states should be far from the code-spaces of stabilizer states. Using results by \cite{Anshu2020CircuitLB}, we prove in \Cref{claim:not-recovered-LB} that if the key is not recovered, then the probability that the inner code syndrome measurements in $\Dec_{\Tilde{Q}_{in}}$ (\Cref{alg:dec_qnm_inner}) all accept decays exponentially with $t$.

\begin{claim}\label{claim:not-recovered-LB}
    If the key is not recovered, then for any set $B$ of $\leq t/b$ inner blocks, the probability that the syndrome measurement in $\Dec_{\Tilde{Q}_{j}}$ accepts for every $j\in B$ is
    \begin{equation}
        \mathbb{P}\bigg(\text{Accept }B|\text{ not recovered}\bigg) \leq \prod_{j\in B} \Tr\bigg[\Pi_{Q_j} \otimes_{i\in B_j}  \Lambda_{i, \Tilde{s}}(\mathbb{I}/2) \bigg] \leq 2^{-\Omega(t\cdot d^2/b^2)}
    \end{equation}
    where we recall that each $Q_i$ is a stabilizer code of block-length $b$ and distance $d$ .
\end{claim}

\textbf{If the classical key is recovered}, then the $t$-wise independent Pauli-one-time-pad is shared between Alice and Bob. To address Bobs decoding, we essentially divide into cases on whether each block is significantly tampered, or whether each tampering channel $\Lambda_i$ is ``sparse" in the Pauli basis. The key technical definition to threshold between these two cases for quantum channels is the following notion of an $\eta$-Pauli qubit channel:

\begin{definition}\label{definition:eta-pauli}
    Let $\Lambda$ be a single qubit channel with Krauss operators $\{K_\mu\}_\mu$, and let $K_\mu = \sum_\sigma c_\sigma^\mu \sigma$ be a Pauli basis decomposition of each $K_\mu$. Then, we refer to $\Lambda_i$ as $\eta$-Pauli if there exists a Pauli $E$ such that

    \begin{equation}
        \sum_\mu |c^\mu_E|^2\geq 1-\eta
    \end{equation}
\end{definition}

Although we only pick $\eta$ later, consider $\eta = \Theta(1/\text{polylog}(k))$ for now.

If there are few $\eta$-Pauli channels in any given block $j$, then, we reason that the approximate quantum error detection properties of each inner code $\Tilde{Q}_j$ (comprised of a stabilizer code and a PMD code) should reject with high probability. In the below, we consider $d^*$ to be the ``pure" distance of the stabilizer code $Q_j$, i.e. the minimum weight of any normalizer element (including the stabilizers themselves). For the binary QRS codes \cite{Grassl_1999} $d^*$ is simply the normal code distance $d$. 

\begin{claim}\label{claim:not-many-eta}
    Assume there exists a subset $B$ of $\geq t/b$ inner $[[b, k_{in}, d^*]]$ stabilizer codes, where for each inner code $j\in B$ at most $(1-\delta)\cdot b$ of the qubit channels $\Lambda_i$, $i\in B_j$ are $\eta$-Pauli. Then, $\Dec_{\Tilde{Q}_{j}}$ accepts for all $j\in B$  with probability bounded by
    \begin{equation}
        \mathbb{P}\bigg(\text{Accept }B|\text{ key recovered}\bigg)  \leq (\epsilon^2_\PMD + (1-\eta)^{\min(d^*, \delta \cdot b)})^{t/b}
    \end{equation}
\end{claim}

Conversely, if there are many $\eta$-Pauli channels in any given block $j$, then we reason that $\Tilde{Q}_j$ should detect the presence of errors without distorting the state.

\begin{claim}\label{claim:many_eta}
    Assume there exists a subset $B$ of $ \geq n-t/b$ inner codes, where for each inner code $j\in B$ at least $b-d^*$ of the qubit channels $\Lambda_i$, $i\in B_j$ are $\eta$-Pauli; assume $\epsilon_\PMD\cdot 2^{8\eta  \cdot b} < 1 /(n\cdot 10^2)$. Then, if $\Dec_{\Tilde{Q}_{j}}$ accepts on every $j\in B$, 
    \begin{equation}
        \|\otimes_{j\in B} [\Pi_{\Tilde{Q}_j} \circ \Lambda_{B_j}\circ \Enc_{\Tilde{Q}_j}] (\rho) - \text{const }\cdot \otimes_{j\in B}\Enc_{\Tilde{Q}_j} (\rho)\|_1\leq  10^2\cdot n\cdot \epsilon_\PMD\cdot 2^{8\eta \cdot b}
    \end{equation}
\end{claim}

\Cref{claim:many_eta} consists essentially of an application of the Gentle measurement lemma to each inner block in $B$, in parallel. For clarity, in the statement above we have omitted writing the channels acting outside $B$ (as they commute with those on $B$).

\textbf{Putting Everything Together}

We can now put everything together. In \Cref{claim:all_together}, we show how to combine all these cases to guarantee that Bob always either recovers a state near the original message, or rejects. Then, in \Cref{corollary:parameter_choices_together} we show that the parameter choices in \Cref{claim:parameter_choices} ensure an exponentially decaying error. 

\begin{claim}\label{claim:all_together}

Assume $\epsilon_\PMD\cdot 2^{8\eta \cdot b} < 1 /(n\cdot 10^2)$ and $t<\frac{1}{2}\cdot d_{out}$. For any state $\psi = \psi_{AR}$ held by Alice and some purification $R$ which does not go through the channel, Bob receives a state close to a convex combination of the original message and rejection:
\begin{equation}
    \|\Dec\circ \Lambda \circ \Enc(\psi) - \rho'\|_1 \leq 100n\cdot \epsilon_\PMD\cdot 2^{8\eta \cdot b} + 2^{-\Omega(t\cdot d^2/b^2)} + 2\cdot (\epsilon^2_\PMD + 2^{-\eta d^*})^{t/b} + \epsilon_{\NM}
\end{equation}

Where $\rho' = c_\Lambda \cdot \psi\otimes\Acc + (1-c_\Lambda)\cdot \bot$, for some real positive constant $c_\Lambda$.
    
\end{claim}

\begin{corollary}\label{corollary:parameter_choices_together}
With the choice of parameters in \Cref{claim:parameter_choices}, $(\Enc, \Dec)$ is a quantum tamper detection code encoding $k$ qubits into $k\cdot (1+O(1/\log k))$ qubits with error $2^{-\Tilde{\Omega}(k)}$.
\end{corollary}

\begin{proof}
 From \Cref{claim:parameter_choices}, we know that $\epsilon_\NM = 2^{-\Tilde{\Omega}(k)}$ and $\epsilon_\PMD \leq 2^{-k/(10\cdot\log^{c+1} k)}$. 
 
 Our choice of inner code has blocklength $b \leq 2 k/\log^c k$, distance and pure distance $d=d^* = \Omega(k/\log^{c+2}k)$, and $b/d = O(\log^2 k)$. We pick $t = b\log k \leq 2 k/\log^{c-1} k \leq \frac{k}{2\log^2 k}\leq \frac{1}{2}d_{out}$, so long as $c>4$ and $k$ is sufficiently large.

 Thus, if we pick $\eta = 1/\log^3 k$, then $2^{8\eta b} \leq 2^{20k/\log^{c+3}k}$ and we fulfill the conditions of \Cref{claim:all_together} since $\epsilon_\PMD\cdot 2^{8\eta b}\leq 2^{-k/(20\cdot\log^{c+1} k)}\ll 1/n$. We conclude that the overall error in statistical distance is 
 \begin{equation}
     \epsilon\leq 2^{-\Tilde{\Omega}(k)} + 2^{-k/(20\cdot\log^{c+1} k)} + 2^{-\Omega(k/\log^{3(c+1) k})}+2^{-\Omega(k/\log^{c+5} k)} \leq 2^{-\Tilde{\Omega}(k)}. 
 \end{equation}

and thereby, $(\Enc, \Dec)$ is $\epsilon$-secure as in \Cref{definition:qa}.
\end{proof}

\subsection{The Proofs}
\label{section:proofs-rate1}

\subsubsection{Two Packing Arguments}

We start with two ``packing arguments", which bound the norms of coefficients in Pauli basis decomposition of Krauss operators of the channels $\Lambda$. Then we move onto the proofs of the previous section. Recall \Cref{definition:eta-pauli} of an $\eta$-Pauli channel. 

\begin{claim}\label{claim:packing-not-eta-pauli}
    If a stabilizer code $Q$ of block-length $b$ and pure distance $d^*$ is acted on by a channel $\otimes_{i\in [b]}\Lambda_i$ where at most $(1-\delta)\cdot b$ of the $\Lambda_i$ are $\eta$-Pauli, then

    \begin{equation}
        \sum_{\mu, F = \otimes_{i\in [b]} F_i \in S(Q)} |c_F^\mu|^2\leq (1-\eta)^{\min(\delta\cdot b, d^*)}
    \end{equation}

    Where each Krauss operator $K^\mu$ of $\Lambda$ has a Pauli basis decomposition given by $\sum_F c_F^\mu F$.
\end{claim}

\begin{proof}
    Let $T\subset [b]$ be any subset of channels which are not $\eta$-Pauli. For every $F = \otimes_{i\in [b]}F_i$, we note that
    \begin{gather}
        \sum_\mu |c_F^\mu|^2 = \prod_{i\in [b]} \bigg( \sum_{\mu_i} |c_{F_i}^{\mu_i}|^2\bigg) \leq \prod_{i\in [b]\setminus T} \bigg( \sum_{\mu_i} |c_{F_i}^{\mu_i}|^2\bigg) \prod_{i\in T} \big( \max_{E_i}\sum_{\mu_i} |c_{E_i}^{\mu_i}|^2\bigg) \leq \\ \leq \prod_{i\in [b]\setminus T} \bigg( \sum_{\mu_i} |c_{F_i}^{\mu_i}|^2\bigg)\cdot  (1-\eta)^{|T|}
    \end{gather}

    By definition of the pure distance $d^*$, once one fixes $n-d^*$ elements of $F$, $F_1\otimes F_2\cdots F_{b-d^*}$, there is only a single ``completion" $F_{n-d^*+1}\otimes \cdots F_n$ such that $\otimes F_i\in S(Q)$. Let $T'\subset T$ be any subset $\min(\delta\cdot b, d^*)$ of the qubit channels which are not $\eta$-Pauli. Then, 
    \begin{gather}
        \sum_{\mu, F = \otimes_{i\in [b]} F_i \in S(Q)} |c_F^\mu|^2 \leq (1-\eta)^{\min(\delta\cdot b, d^*)}  \sum_{F\in S(Q)}\prod_{i\in [b]\setminus T} \bigg( \sum_{\mu_i} |c_{F_i}^{\mu_i}|^2\bigg) \leq \\ \leq(1-\eta)^{\min(\delta\cdot b, d^*)}  \sum_{F_1, \cdots F_{[b]\setminus T}}\prod_{i\in [b]\setminus T} \bigg( \sum_{\mu_i} |c_{F_i}^{\mu_i}|^2\bigg) = \\= (1-\eta)^{\min(\delta\cdot b, d^*)} \prod_{i\in [b]\setminus T}\bigg( \sum_{\mu_i, F_i} |c^{\mu_i}_{F_i}|^2\bigg) =  (1-\eta)^{\min(\delta\cdot b, d^*)} 
    \end{gather}
\end{proof}

\begin{claim}\label{claim:packing-eta-pauli}
    Let $\Tilde{Q}$ be the composition of an $\epsilon$-$\PMD$ and a stabilizer code of block-length $b$ and ``pure'' distance $d^*$. If $\Tilde{Q}$ is acted on by a channel $\Lambda = \otimes \Lambda_i$ where at least $b-d^*$ of the $\Lambda_i$ are $\eta$-Pauli, then 
    \begin{equation}
        \sum_\mu \big(\sum_{F\in N(Q)} |c_F^\mu|\big)^2\leq 2^{8\eta\cdot b}
    \end{equation}
    Where each Krauss operator $K^\mu$ of $\Lambda$ has a Pauli basis decomposition given by $\sum_F c_F^\mu F$.
\end{claim}

\begin{proof}

The proof of this claim follows from a ``packing" argument. Fix a Krauss operator $\mu = (\mu_1, \cdots, \mu_b)$, $K_\mu = \sum_\sigma c_\sigma^\mu \sigma$, and let us consider the Pauli's $\sigma\in N(Q)$. After fixing the first $b-d^*$ single qubit Paulis $\sigma_1, \sigma_2\cdots \sigma_{b-d^*}$, the definition of ``pure" distance ensures there only exists a single ``completion" $\sigma_{b-d^*+1}\otimes \cdots \sigma_b$ such that $\sigma\in N(Q)$. Thus, 
    \begin{gather}
        \sum_{\sigma\in N(Q)} |c_\sigma^\mu | = \prod_{i\in [b-d^*]}\bigg( \sum_{\sigma_i} |c^{\mu_i}_{\sigma_i}|\bigg) \times \prod_{i\in [b-d^*,\cdots, b]}\max_{\sigma_i}|c^{\mu_i}_{\sigma_i}| \leq \\ \leq  \prod_{i\in [b-d^*]}\bigg( \sum_{\sigma_i} |c^{\mu_i}_{\sigma_i}|\bigg) \times \prod_{i\in [b-d^*,\cdots, b]}  \bigg(\sum_{\sigma_i}|c^{\mu_i}_{\sigma_i}|^2\bigg)^{1/2}
    \end{gather}

    To conclude, if for $i\in [b-d^*]$, $\Lambda_i$ is $\eta$-Pauli for some Pauli $E_i$
    \begin{gather}
        \sum_{\mu_i}\bigg(\sum_{\sigma_i} |c^{\mu_i}_{\sigma_i}|\bigg)^2 = \sum_{\mu_i} |c_{E_i}^{\mu_i}|^2 + 2\sum_{\mu_i} |c_{E_i}^{\mu_i}| \cdot \bigg(\sum_{\sigma_i\neq E_i} |c^{\mu_i}_{\sigma_i}|\bigg) +\sum_{\mu_i}
        \bigg(\sum_{\sigma_i\neq E_i} |c^{\mu_i}_{\sigma_i}|\bigg)^2 \leq  \\  \leq 1 + 2\cdot 2\bigg(\sum_{\sigma_i\neq E_i,\mu_i} |c^{\mu_i}_{\sigma_i}|^2\bigg)^{1/2}+ 2^2\bigg(\sum_{\sigma_i\neq E_i,\mu_i} |c^{\mu_i}_{\sigma_i}|^2\bigg) \leq (1+2\eta)^2\leq 2^{8\eta }
    \end{gather}

    Where between lines we applied the Cauchy-Schwartz inequality twice, that $\sum_{\mu_i, \sigma_i}|c^{\mu_i}_{\sigma_i}|^2 = 1$, and finally $1+x\leq 2^{2x}$. Plugging this into the packing argument gives us the claim. 
\end{proof}

\subsubsection{If the key is recovered}

\begin{proof}

[of \Cref{claim:not-many-eta}]
    We note that within each single block $j\in B$, the marginal distribution on whether the code-space projection $\Pi_{\Tilde{Q}_j}$ accepts is given by
    \begin{gather}
        \mathbb{P}\bigg(\text{Accept }B_j|\text{ key recovered}\bigg) = \Tr_{B_j}\bigg[\Pi_{\Tilde{Q}_j} \mathbb{E}_s \bigg[\bigg(\otimes_{i\in B_j}(P^{G(s)})^\dagger_i\circ \Lambda_i \circ P^{G(s)}_i \bigg) \big(\Enc_{\Tilde{Q}_j}(\psi_j)\big)\bigg]\bigg] \\
        =\Tr_{B_j}\bigg[\Pi_{\Tilde{Q}_j} \sum_{F_j, \mu_j} |c_{F_j}^{\mu_j}|^2\cdot F\cdot \Enc_{\Tilde{Q}_j}(\psi_j) \cdot F^\dagger\bigg] 
    \end{gather}
    Where we applied a Pauli-twirl (\cref{fact:twirl}), with the guarantee that the substring $G(s)_{B_j}$ is uniformly random (since $b\leq t$). In the above, $c_{F_j}^{\mu_j}$ corresponds to the coefficient corresponding to the $b$-qubit Pauli $F_j$ in a Pauli basis decomposition of the $\mu_j$th Krauss operator of $\Lambda_j$

\begin{equation}
\leq \sum_{\mu_j, F_j\in S(Q_j)}|c_{F_j}^{\mu_j}|^2+\epsilon^2_\PMD \sum_{F\in N(Q_j)-S(Q_j)}|c_{F_j}^{\mu_j}|^2 \leq \sum_{F\in S(Q_j)}|c_{F_j}^{\mu_j}|^2+\epsilon^2_\PMD \leq \epsilon^2_\PMD + (1-\eta)^{\min(d^*, \delta \cdot b)}
\end{equation}

   Where we used that $\Tilde{Q}_j$ is the composition of a stabilizer code with a $\epsilon_\PMD$-$\PMD$, the fact that $\sum_{F_j, \mu_j}|c_{F_j}^{\mu_j}|^2=1$, and \Cref{claim:packing-not-eta-pauli}. Since the pseudorandom one-time-pad is $t$-wise independent, we have 

   \begin{equation}
       \mathbb{P}\bigg(B_1, B_2\cdots B_{t/b} \text{ Accept}|\text{ key recovered}\bigg) = (\epsilon^2_\PMD + (1-\eta)^{\min(d^*, \delta \cdot b)})^{t/b}
   \end{equation}
\end{proof}

\begin{proof}

[of \Cref{claim:many_eta}]
    We start by considering a single block $j\in B$. If $K^\mu_j = \sum_{F} c_{F, j}^{\mu_j}$ is a Pauli basis decomposition of each Krauss operator $K^\mu_j$ of $\Lambda_j$, let $a_F = \sum_{\sigma  \in S(Q) }c_{F\sigma, j}$ be the sum of all the coefficients $c_{F\sigma}$ of operators equivalent up to a stabilizer to $F$. We can group the action of $\Pi_{\Tilde{Q}}$ and $\Lambda$ based on cosets of $N(Q)$:

    \begin{equation}
        \Pi_{\Tilde{Q}_j} \circ  \Lambda_j \circ \Enc_{\Tilde{Q}_j}(\psi) = \sum_{\mu} \sum_{E, F\in N(Q_j)\setminus S(Q_j)}a_E^\mu (a_F^\mu)^*\cdot \Pi_{\Tilde{Q}_j} E \Enc_{\Tilde{Q}_j}(\psi) F^\dagger \Pi_{\Tilde{Q}_j} 
    \end{equation}

    Let $\alpha_j = \sum_\mu |a^\mu_{\mathbb{I}}|^2$. Using $\|\Pi_{\Tilde{Q}} E \Pi_{\Tilde{Q}}\|_\infty =0$ if $E\notin N(Q)$ and $\|\Pi_{\Tilde{Q}} E \Pi_{\Tilde{Q}}\|_\infty \leq \epsilon_\PMD$ if $E\in N(Q)$ but not $S(Q)$, 
    \begin{gather}
\bigg\|\Pi_{\Tilde{Q}_j} \circ  \Lambda_j \circ \Enc_{\Tilde{Q}_j}(\psi) - \alpha_j\cdot \Enc_{\Tilde{Q}_j}(\psi)\bigg\|_1 \leq \\ \leq 2\cdot \epsilon_\PMD\cdot  \sum_\mu |a_\mathbb{I}^\mu|\cdot \sum_{\substack{E\in N(Q_j)\setminus S(Q_j) \\ E\neq \mathbb{I}}} |a_E^\mu| + \epsilon^2_\PMD \bigg( \sum_{\substack{E\in N(Q_j)\setminus S(Q_j) \\ E\neq \mathbb{I}}} |a_E^\mu| \bigg)^2 \\
\leq \epsilon_\PMD \sum_\mu \bigg( |a_\mathbb{I}^\mu|^2 + \bigg( \sum_{\substack{E\in N(Q_j)\setminus S(Q_j) \\ E\neq \mathbb{I}}} |a_E^\mu| \bigg)^2\bigg) + \epsilon^2_\PMD \bigg( \sum_{\substack{E\in N(Q_j)\setminus S(Q_j) \\ E\neq \mathbb{I}}} |a_E^\mu| \bigg)^2 \leq \\ \leq \epsilon_\PMD\cdot \alpha_j + 2\cdot\epsilon_\PMD\cdot \sum_\mu \bigg(\sum_{\substack{E\in N(Q_j)\setminus S(Q_j) \\ E\neq \mathbb{I}}} |a_E^\mu| \bigg)^2 \leq 4\cdot\epsilon_\PMD\cdot 2^{8\eta  \cdot b} 
    \end{gather}

    Where in the last line we used \Cref{claim:packing-eta-pauli}, and by assumption that in block $j$, at least $b-d^*$ of the qubit channels are $\eta$-Pauli. By a simple hybrid argument, we can bootstrap this statement to all the blocks in $B$
\begin{gather}
        \|\otimes_{j\in B} \Pi_{\Tilde{Q}_j} \circ \Lambda_{B_j}\circ \Enc_{\Tilde{Q}_j} (\psi) - \bigg( \prod_j \alpha_j\bigg)\cdot \otimes_{j\in B}\Enc_{\Tilde{Q}_j} (\psi)\|_1\leq \\ \leq \sum_{\ell \in [|B|]} \bigg(\prod_{j\in [\ell -1]} \alpha_j\bigg) \cdot  \|\otimes_{j\in [\ell]} \Pi_{\Tilde{Q}_j} \circ \Lambda_{B_j}\circ \Enc_{\Tilde{Q}_j} (\psi) - \alpha_\ell\cdot \otimes_{j\in [\ell-1]} \Pi_{\Tilde{Q}_j} \circ \Lambda_{B_j}\circ \Enc_{\Tilde{Q}_j} (\psi) \|_1 \\ \leq n\cdot \bigg( 4\cdot\epsilon_\PMD\cdot 2^{8\eta  \cdot b} \bigg) \bigg( 1+4\cdot\epsilon_\PMD\cdot 2^{8\eta  \cdot b} \bigg)^n \leq 10^2\cdot n\cdot \epsilon_\PMD\cdot 2^{8\eta \cdot b}
\end{gather}

so long as $\epsilon_\PMD\cdot 2^{8\eta \cdot b} < 1 /(n\cdot 10^2)$, where we used $\alpha_j\leq 1+4\cdot\epsilon_\PMD\cdot 2^{8\eta  \cdot b}$.
\end{proof}

\subsubsection{If the key is not recovered}

\begin{proof}

[of \Cref{claim:not-recovered-LB}] By the $t\geq (|B|\cdot b)$-wise independence, we note that the syndrome measurements on each block of $B$ are independent events. Moreover, on average over the key,  by \cref{fact:1design} the density matrix on $B$ is a tensor product of mixed states $\otimes_{j\in [B]} \otimes_{i\in B_j}  \Lambda_{i, \Tilde{s}}(\mathbb{I}/2)$. By an averaging argument,  $p\bigg(\text{Accept }B|\text{ not recovered}\bigg)$ is maximized when each $\Lambda_{i, \Tilde{s}}(\mathbb{I}/2)$ is a pure state $\otimes_{i\in [n]}\ket{\phi_i}$. 

However, each $Q_i$ is a stabilizer code, and thus its code-space should not have high fidelity with any product state. Again, we use the circuit lower bound by \cite{Anshu2020CircuitLB}, restated in \cref{lemma:clb}, to prove that each inner code $Q_i$ for $j\in B$ rejects the syndrome measurement with high probability:
\begin{equation}
    \Tr\big[ \Pi_{Q_i} \otimes_{j\in B_i}\Lambda_{j, \Tilde{s}}(\mathbb{I}/2)\big] \leq 2^{-d^2/(2^{12}\cdot b)}
\end{equation}

Thus, all the blocks in $B$ accept with probability at most
\begin{equation}
    p\bigg(\text{Accept }B|\text{ not recovered}\bigg) = \prod_{j\in B} \Tr\big[ \Pi_{Q_j} \otimes_{i\in B_j}\Lambda_{i, \Tilde{s}}(\mathbb{I}/2)\big] \leq 2^{-t\cdot d^2/(2^{12}\cdot b^2)}
\end{equation}

as desired. 
\end{proof}

\subsubsection{Putting Everything Together}

\begin{proof}

[of \Cref{claim:all_together}] Let $\rho=\psi_{AR}$ be any message state in Alice's possession, with some purification $R$ which does not go through the channel. For notational convenience, let us consider first applying the channels on the classical register $C$, $\otimes_{i\in C} \Lambda_i $. From \Cref{claim:classical-nm},
\begin{equation}
   \| \big( (\Dec_\NM \circ \otimes_{i\in C} \Lambda_i) \otimes \mathbb{I}_{\Tilde{Q}}\big)\Enc_{}(\psi) - \sum_{\Tilde{s}, s}\frac{p_{\Lambda, s}(\Tilde{s})}{|K|} \ketbra{\Tilde{s}}\otimes \big(P^{G(s)} \circ \Enc_{\Tilde{Q}} (\psi) \big) \|_1\leq \epsilon_{\NM}
\end{equation}
Where $p_{\Lambda, k}(\Tilde{k})$ corresponds to the convex combination $p_{\Lambda, s}(\Tilde{s}) = p_{\Lambda, \text{Same}}\cdot  \delta_{s, \Tilde{s}} + p_{\Lambda, \bot}\cdot \delta_{\bot, \Tilde{s}} + (1-p_{\Lambda, \text{Same}}-p_{\Lambda, \bot})\cdot \mathcal{D}^\Lambda(\Tilde{s})$, for some positive real parameters $p_{\Lambda, \text{Same}}, p_{\Lambda, \bot}$. 

 The case where the key is not recovered corresponds to $\Tilde{s}\leftarrow \mathcal{D}^\Lambda$. From \Cref{claim:not-recovered-LB}, after applying the block-wise syndrome measurements in \Cref{alg:dec_qnm_inner}, this case must have high fidelity with $\bot$ (rejection): for all $\Tilde{s}$, 
\begin{gather}
   \bigg\| \Dec_{\Tilde{Q}}\circ (P^{G(\Tilde{s})})^\dagger\circ  \mathbb{E}_{s} \bigg(P^{G(s)} \circ \Enc_{\Tilde{Q}}(\psi)\bigg) - \bot\bigg\|_1 \leq 2\big(1-\mathbb{P}[\text{reject}|\text{key not recovered}]\big) \leq \\\leq 2\cdot \mathbb{P}\bigg[\text{All Inner Syndrome Meas. Accept}\big| \text{key not recovered}\bigg]\leq 2^{-\Omega(t\cdot d^2/b^2)}
\end{gather}

If the key is recovered, we divide into two cases. From \Cref{claim:not-many-eta}, if there exist a subset of $\geq t/b$ inner codes where $\geq d^*$ of their qubit channels is not $\eta$-Pauli, then the probability all the inner code syndrome and PMD measurements accept decays exponentially:
\begin{gather}
   \bigg\| \mathbb{E}_{s}\Dec_{\Tilde{Q}}\circ (P^{G(s)})^\dagger \circ \Lambda \circ P^{G(s)} \circ \Enc_{\Tilde{Q}}(\psi) - \bot\bigg\|_1 \leq 2\cdot \mathbb{P}\bigg[\Dec_{\Tilde{Q}_j}\text{ accepts for all $j$}\big| \text{key recovered}\bigg] \\
   \leq 2\cdot (\epsilon^2_\PMD + 2^{-\eta d^*})^{t/b}
\end{gather}

The second case is if there exists a subset $B$ of $\geq n - t/b$ inner codes where $\geq b-d^*$ of their qubit-wise channels are $\eta$-Pauli. By \Cref{claim:many_eta}, if the inner code error detection steps accept on $B$, then, the ``effective" channel corrupts only few inner code blocks. Note that the channel still factorizes for each choice of $s$, $\Lambda_s=(P^{G(s)})^\dagger \circ \Lambda\circ (P^{G(s)}) = \otimes_{i}\Lambda_{i, s} = \otimes_i\big[(P^{G(s)})^\dagger_i \circ \Lambda_i\circ (P^{G(s)})_i\big]$, and moreover, since $(P^{G(s)})_i$ is simply a single qubit Pauli, $\Lambda_{i, s}$ and $\Lambda_{i}$ have the same Pauli ``Sparse-ness". Thus, if $\Pi_{\Tilde{Q}_j}$ corresponds to the accepting branch of the measurements in $\Dec_{\Tilde{Q}_j}$ on $\Tilde{Q}_j$, for $j\in B$:
\begin{equation}
        \|(\otimes_{j\in B} \Pi_{\Tilde{Q}_j}) \circ \Lambda_s \circ \Enc_{A\rightarrow \Tilde{Q}} (\psi_{AR}) - \text{const }\cdot (\otimes_{\substack{j\notin B \\ i\in B_j}}\Lambda_{i, s})\circ \Enc_{A\rightarrow \Tilde{Q}} (\psi_{AR})\|_1\leq  10^2\cdot n\cdot \epsilon_\PMD\cdot 2^{8\eta \cdot b}
\end{equation}

The effective channel (CP map) $\Lambda_{\Bar{B}} = \otimes_{j\in \Bar{B}} \Pi_{\Tilde{Q}_j} \circ \otimes_{\substack{i\in \Bar{B}_j}}\Lambda_{i, s}$ is thus close to an operator supported only on $\Bar{B}$. $\Bar{B}$ consists of at most $t/b$ inner code blocks, and thus in turn at most $\frac{t}{b}\times k_{\PMD} \leq t$ symbols of the outer code. Since $t < \frac{1}{2}d_{out}$, the outer code can detect the effective channel exactly by measuring its syndrome: 
\begin{equation}
    \Dec_{Q_{out}}\circ \bigg(\otimes_j \Enc^\dagger_{\Tilde{Q}_j}\bigg)\circ \bigg(\Lambda_{\Bar{B}}\otimes \mathbb{I}\bigg)\circ \Enc_{A\rightarrow \Tilde{Q}} (\psi_{AR}) = \beta_{acc} \cdot \psi_{AR} +  \Rej_A \otimes \sigma_{R} \cdot \beta_{rej}
\end{equation}

We conclude, for all qubit-wise $\Lambda$, Bob receives a density matrix near $\rho' = c_\Lambda \cdot \psi_{AR} + (1-c_\Lambda)\cdot \bot$ which is a convex combination of the original message and rejection. The approximation error is given by:
\begin{equation}
    \|[\Dec_{}\circ \Lambda \circ \Enc_{}\otimes \mathbb{I}_{R}](\psi_{AR}) - \rho'\|_1 \leq 100n\cdot \epsilon_\PMD\cdot 2^{8\eta  \cdot b} + 2^{-\Omega(t\cdot d^2/b^2)} + 2\cdot (\epsilon^2_\PMD + 2^{-\eta d^*})^{t/b} + \epsilon_{\NM}
\end{equation}

\end{proof}


\end{document}